%% file: main_arxiv.tex
\documentclass[runningheads]{llncs}

\usepackage{amsmath}
\usepackage{amssymb}
\usepackage{graphicx}
\usepackage{fink}
\usepackage[table]{xcolor}
\usepackage[vlined]{algorithm2e}
\usepackage{mathpartir}
\usepackage{stmaryrd}
\usepackage[font=footnotesize]{caption} \DeclareCaptionType{copyrightbox}
\usepackage{verbatimbox}
\usepackage{subcaption}
\usepackage[T1]{fontenc}

\usepackage{hyperref}

\newcommand{\Prg}{\mathcal{A}}

\newcommand{\spacer}{\textsc{Spacer}\xspace}
\newcommand{\recmc}{\textsc{RecMC}\xspace}
\newcommand{\bdsafety}{\textsc{BndSafety}\xspace}
\newcommand{\bndsafety}{\textsc{BndSafety}\xspace}

\newcommand{\ie}{{\em i.e.},\xspace}
\newcommand{\viz}{{\em viz.},\xspace}
\newcommand{\eg}{e.g.,\xspace}
\newcommand{\st}{\cdot}

\newcommand{\subheading}[1]{\noindent\textbf{#1.}}
\newcommand{\subheadingnodot}[1]{\noindent\textbf{#1}}

\newcommand{\bigand}{\bigwedge}
\newcommand{\bigor}{\bigvee}
\renewcommand{\implies}{\Rightarrow}

\newcommand{\union}{\cup}

\newcommand{\htriple}[3]{\{{#1}\} ~{#2}~ \{{#3}\}}

\renewcommand{\sum}[1]{\Sigma_{#1}}

\newcommand{\sem}[1]{\llbracket {#1} \rrbracket}

\newcommand{\iparams}[1]{\vec{\iota}_{#1}}
\newcommand{\oparams}[1]{\vec{o}_{#1}}
\newcommand{\formals}[1]{\vec{v}_{#1}}
\newcommand{\locals}[1]{\vec{\ell}_{#1}}
\newcommand{\lits}[1]{\text{\emph{lits}}({#1})}

\newcommand{\body}[1]{\beta_{#1}}
\newcommand{\paths}[1]{\mathit{Paths}({#1})}

\newcommand{\projf}[1]{\textnormal{\emph{Proj}}_{#1}}
\newcommand{\lraprojf}[1]{\textnormal{\emph{LRAProj}}_{#1}}

\renewcommand{\vec}{\overline}

\newcommand{\thry}{\mathit{Th}}

\newcommand{\expansion}[3]{{#1}\{{#2} \mapsto {#3}\}}
\newcommand{\sig}{\mathcal{S}}

\providecommand{\DontPrintSemicolon}{\dontprintsemicolon}

\renewcommand{\paragraph}[1]{\noindent \emph{#1}}

\spnewtheorem*{theorem*}{Theorem}{\bfseries}{\upshape}

\definecolor{midgrey}{rgb}{0.3,0.3,0.3}
\definecolor{darkred}{rgb}{0.7,0.1,0.1}

\title{SMT-based Model Checking for\\Recursive Programs
\thanks{This paper is originally published by Springer-Verlag as part of the
proceedings of CAV 2014. The final publication is available at
link.springer.com.}}

\titlerunning{SMT-based Model Checking for Recursive Programs}

\author{Anvesh Komuravelli \and Arie Gurfinkel \and Sagar Chaki}
\authorrunning{A.~Komuravelli et al.}
\institute{Carnegie Mellon University, Pittsburgh,
PA, USA
}

\date{}

\begin{document}
\maketitle

\input{abstract}
\input{intro_arxiv}

\input {overview}
\input{prelims}
\input{algo_arxiv}
\input{la_arxiv}
\input{results}
\input{related}
\input{conclusion}
\newpage
\section*{Acknowledgment}
We thank Edmund M. Clarke and Nikolaj Bj{\o}rner for many helpful
discussions. Our definition of MBP is based on the idea of projected
implicants co-developed with Nikolaj. We thank Cesare Tinelli and the
anonymous reviewers for insightful comments. This research was
sponsored by the National Science Foundation grants no.~DMS1068829,
CNS0926181 and CNS0931985, the GSRC under contract no.~1041377, the
Semiconductor Research Corporation under contract no.~2005TJ1366, the
Office of Naval Research under award no.~N000141010188 and the
CMU-Portugal Program.  This material is based upon work funded and
supported by the Department of Defense under Contract
No. FA8721-05-C-0003 with Carnegie Mellon University for the operation
of the Software Engineering Institute, a federally funded research and
development center.
    Any opinions, findings and conclusions or recommendations
    expressed in this material are those of the author(s) and do not
    necessarily reflect the views of the United States Department of
    Defense.
    This material has been approved for public release and unlimited
    distribution.
    DM-0000973. 

%\bibliography{refs}
%\bibliographystyle{abbrv}

\input{main_arxiv.bbl}
\input{appendix}

\end{document}

%% file: abstract.tex
\begin{abstract}
  We present an SMT-based symbolic model checking algorithm for safety
  verification of recursive programs.
  The algorithm is modular and analyzes procedures individually.
  Unlike other SMT-based approaches, it maintains both
  \emph{over-} and \emph{under-approximations} of procedure
  summaries. Under-approximations are used to analyze procedure calls without
  inlining. Over-approximations are used to block
  infeasible counterexamples and detect convergence to a proof.
  We show that for programs and properties over a decidable theory, the algorithm
  is guaranteed to find a counterexample, if one exists. However, efficiency
  depends on an oracle for quantifier elimination (QE). % Given such an
  % oracle, the queries for individual procedures and the inferred
  % approximations of the summaries depend only on the variables in
  % scope.
  % This makes the algorithm significantly more efficient than
  % existing
  % SMT-based algorithms which unroll the call-graph.
%
  For Boolean Programs, the algorithm is a polynomial decision
  procedure, matching the worst-case bounds of the best BDD-based
  algorithms.
  For Linear Arithmetic (integers and rationals), we give an efficient
  instantiation of the algorithm by applying QE \emph{lazily}.  We use
  existing interpolation techniques to over-approximate QE and
  introduce \emph{Model Based Projection} to under-approximate QE.
  Empirical evaluation on SV-COMP benchmarks shows that our algorithm
  improves significantly on the state-of-the-art.
\end{abstract}

%%% Local Variables:
%%% mode: latex
%%% TeX-master: "main"
%%% End:

%% file: intro_arxiv.tex
%%%
\newcommand{\D}{\mathcal{D}}
\newcommand{\B}{\mathcal{B}}
\section{Introduction}
\label{sec:introduction}

We are interested in the problem of \emph{safety} of recursive programs, \ie
deciding whether an assertion always holds. The first step in Software
Model Checking is to approximate the input program by a program model where the
program operations are terms in a first-order theory $\D$. Many program models
exist today, \eg \emph{Boolean Programs}~\cite{bebop} of SLAM~\cite{slam},
\textsc{Goto} programs of CBMC~\cite{cbmc}, \textsc{BoogiePL} of
\textsc{Boogie}~\cite{boogie}, and, indirectly, internal representations of many
tools such as \textsc{UFO}~\cite{ufo}, \textsc{HSF}~\cite{hsfc}, etc.  Given a
safety property and a program model over $\D$, it is possible to analyze bounded
executions using an oracle for \emph{Satisfiability Modulo Theories} (SMT) for
$\D$.  However, in the presence of unbounded recursion, safety is undecidable in
general. Throughout this paper, we assume that procedures cannot be passed as
parameters.

%\emph{Boolean Programs}~\cite{bebop} is a model for software, widely used to
%approximate imperative programs in Software Model Checking. It provides a simple
%way to model control-flow and procedure behavior including local
%variables, unbounded stack and unbounded recursion. It is an instance of the
%more general model of procedural programs with statements in a
%decidable theory. It is possible to analyze bounded executions of such a model
%using an oracle for \emph{Satisfiability Modulo Theories} (SMT). Examples of
%such models include \textsc{Goto} programs of \textsc{CBMC}~\cite{cbmc}, \textsc{BoogiePL}
%of \textsc{Boogie}~\cite{boogie} and, indirectly, internal
%representations of many tools, such as \textsc{UFO}~\cite{ufo},
%\textsc{HSF}~\cite{hsfc}, etc. Throughout this paper, we assume that procedures
%cannot be used as parameters.

There exist several program models where safety is efficiently
decidable\footnote{This is no longer true when we allow procedures as
parameters~\cite{ed_jacm}.}, \eg Boolean Programs with unbounded recursion and
the unbounded use of stack~\cite{rhs,bebop}. The general observation behind
these algorithms is that one can \emph{summarize} the input-output behavior of a
procedure. A summary of a procedure is an input-output relation describing what
is currently known about its behavior. Thus, a summary can be used to analyze a
procedure call without inlining or analyzing the body of the
callee~\cite{ed_paper,sharir81}.  For a Boolean Program, the number of states is
finite and hence, a summary can only be updated finitely many times. This
observation led to a number of efficient algorithms that are polynomial in the
number of states, \eg the RHS framework~\cite{rhs}, recursive
state machines~\cite{pds}, and symbolic BDD-based algorithms of
\textsc{Bebop}~\cite{bebop} and \textsc{Moped}~\cite{moped}. When safety is
undecidable (\eg when $\D$ is Linear Rational Arithmetic (LRA) or Linear Integer
Arithmetic (LIA)), several existing software
model checkers work by iteratively obtaining Boolean Program abstractions using
Predicate Abstraction~\cite{cegar,slam}. In this paper, we are interested in an
alternative algorithm that works directly on the original program model without
an explicit step of Boolean abstraction. Despite the undecidability, we are
interested in an algorithm that is guaranteed to find a counterexample to
safety, if one exists.

Several algorithms have been recently proposed for verifying recursive
programs without predicate abstraction. Notable examples are
\textsc{Whale}~\cite{whale}, HSF~\cite{hsfc}, GPDR~\cite{gpdr},
Ultimate Automizer~\cite{DBLP:conf/tacas/HeizmannCDEHLNSP13,nested_itp} and Duality~\cite{duality}.
With the exception of GPDR, these algorithms are
based on a combination of Bounded Model Checking (BMC)~\cite{bmc} and
Craig Interpolation~\cite{craig}. First, they use an SMT-solver to check for a bounded
counterexample, where the bound is on the depth of the call stack (\ie
the number of nested procedure calls). Second, they use (tree)
interpolation to over-approximate procedure summaries. This is
repeated with increasing values of the bound until a counterexample is found or the
approximate summaries are inductive. The reduction to BMC ensures that
the algorithms are guaranteed to find a counterexample. However, the size of the
SMT instance grows exponentially with the bound on the call-stack (\ie linear in the size of
the call tree). Therefore, for Boolean Programs, these algorithms are at least
worst-case exponential in the number of states.

On the other hand, GPDR follows the approach of IC3~\cite{ic3} by solving
BMC incrementally without unrolling the call-graph. Interpolation is used to
over-approximate summaries and caching is used to indirectly
under-approximate them.
%GPDR is presented as a set of rules applied non-deterministically.
For some configurations, GPDR is worst-case polynomial for Boolean Programs. However, even for LRA, GPDR
might fail to find a counterexample.
\footnote{See appendix for an example.}

In this paper, we introduce $\recmc$, the first SMT-based algorithm for model
checking safety of recursive programs that is worst-case polynomial (in the
number of states) for Boolean Programs while being a co-semidecision procedure
for programs over decidable theories (see Section~\ref{sec:algo}). Our main
insight is to maintain not only over-approximations of procedure summaries
(which we call \emph{summary facts}), but also their under-approximations (which
we call \emph{reachability facts}). While summary facts are used to block
spurious counterexamples, reachability facts are used to analyze a procedure
call without inlining or analyzing the body of the callee. Our use of
reachability facts is similar to that of \emph{summary edges} of the
RHS~\cite{rhs} algorithm. This explains our complexity result for Boolean
Programs. However, our summary facts make an important difference. While the use of
summary facts is an interesting heuristic for Boolean Programs that does not
improve the worst-case complexity, it is crucial for richer theories.
%for convergence for programs over richer theories.

Almost every step of \recmc results in existential quantification of variables.
\recmc tries to eliminate these variables, as otherwise, they would accumulate
and the size of an inferred reachability fact, for example, grows exponentially
in the bound on the call-stack. But, a na\"{i}ve use of quantifier elimination
(QE) is expensive.  Instead, we develop an alternative approach that
under-approximates QE. However, obtaining arbitrary under-approximations can
lead to divergence of the algorithm. We introduce the concept of \emph{Model
Based Projection} (MBP), for \emph{covering} $\exists \vec{x} \st
\varphi(\vec{x},\vec{y})$ by \emph{finitely-many} quantifier-free
under-approximations obtained using models of $\varphi(\vec{x},\vec{y})$. We
developed efficient MBPs (see Section~\ref{sec:lazy_qe}) for Linear Arithmetic
based on the QE methods by Loos-Weispfenning~\cite{lw} for LRA and
Cooper~\cite{cooper} for LIA. We use MBP to under-approximate reachability facts
in $\recmc$. In the best case, only a partial under-approximation is needed and
a complete quantifier elimination can be avoided.

%In general, $\recmc$ maintains formulas with existentially quantified variables.
%Unless these variables are eliminated, they accumulate and the size of an inferred reachability
%fact, for example, grows exponentially in the bound on the call-stack.
%Using quantifier elimination (QE) is expensive.
%Instead, we develop an alternative approach that under-approximates QE for the special
%case of Linear Arithmetic (LRA and LIA).
%We show that there is an efficient \emph{finite-image} \emph{Model-Based Projection} (MBP)
%function that, given a model $M$ of a formula $\varphi(\vec{x},\vec{y})$ returns an
%under-approximation $\psi$ of $\exists \vec{x} \st \varphi(\vec{x},\vec{y})$
%that \emph{covers} $M$ in linear time. Our construction (see Section~\ref{sec:lazy_qe}) is
%based on the QE methods by Loos-Weispfenning~\cite{lw} for LRA and
%Cooper~\cite{cooper} for LIA. We use MBP to under-approximate reachability
%facts in $\recmc$. In the best case, only a partial under-approximation
%is needed and a complete quantifier \mbox{elimination can be avoided.}

We have implemented $\recmc$ as part of our tool \spacer using the framework of
Z3~\cite{z3} and evaluated it on 799 benchmarks
from SV-COMP~\cite{svcomp14}. \spacer significantly
outperforms the implementation of GPDR in Z3 (see Section~\ref{sec:results}).

In summary, our contributions are: (a) an efficient
SMT-based algorithm for model checking recursive programs, that
analyzes procedures individually using under- and over-approximations
of procedure summaries, (b) MBP functions
for under-approximating quantifier elimination for LRA and LIA,
(c) a new, complete algorithm for Boolean Programs,
with complexity polynomial in the number of states, similar to the
best known method~\cite{bebop}, and (d) an implementation and
an empirical evaluation of the approach.

%%% Local Variables:
%%% mode: latex
%%% TeX-master: "main"
%%% End:

%% file: overview.tex
\section{Overview}
\label{sec:overview}

\begin{figure}[t]
\centering
\includegraphics[scale=1]{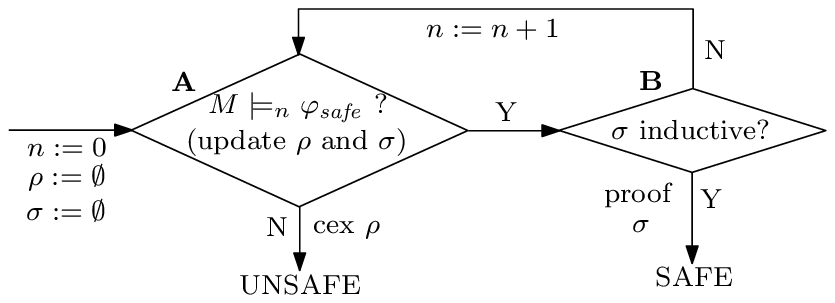}
\caption{Flow of the algorithm \recmc to check if $M \models \varphi_{\mathit{safe}}$.}
\label{fig:recmc}
\end{figure}

In this section, we give an overview of $\recmc$ and illustrate it on
an example. Let $\Prg$ be a recursive program. For simplicity of presentation, assume no loops,
no global variables and that arguments are passed by reference. Let
$P(\vec{v}) \in \Prg$ be a procedure with parameters $\vec{v}$
and let $\vec{v}_0$ be fresh variables not appearing in $P$ with $|\vec{v}| = |\vec{v}_0|$. A safety property for
$P$ is an assertion $\varphi(\vec{v}_0,\vec{v})$. We say that $P$ satisfies
$\varphi$, denoted $P(\vec{v}) \models \varphi(\vec{v}_0,\vec{v})$, iff the Hoare-triple
$\htriple{\vec{v}=\vec{v}_0}{P(\vec{v})}{\varphi(\vec{v}_0,\vec{v})}$
is valid. Note that every Hoare-triple corresponds to a safety property in this
sense, as shown by Clarke~\cite{ed_paper}, using a \emph{Rule of Adaptation}.
Given a safety property $\varphi$ and a natural number $n \ge 0$, the problem of
\emph{bounded safety} is to determine whether all executions of $P$
using a call-stack bounded by $n$ satisfy $\varphi$. We use $P(\vec{v})
\models_n \varphi(\vec{v}_0,\vec{v})$ to denote bounded safety.

The key steps of $\recmc$ are shown in Fig.~\ref{fig:recmc}. \recmc decides
safety for the main procedure $M$ of $\Prg$.
%Given a program $\Prg$ with a main procedure $M \in \Prg$ and a formula
%$\varphi$, $\recmc$ decides whether $M \models \varphi$.
$\recmc$ maintains two \emph{assertion maps} $\rho$ and $\sigma$.
The \emph{reachability} map $\rho$ maps each procedure $P(\vec{v}) \in \Prg$ to
a set of assertions over $\vec{v}_0 \cup \vec{v}$ that under-approximate
its behavior. Similarly, the \emph{summary} map $\sigma$ maps a procedure $P$ to
a set of assertions that over-approximate its behavior. Given $P$, the maps are
partitioned according to the bound on the call-stack. That is, if
$\delta(\vec{v}_0,\vec{v}) \in \rho(P, n)$ for $n \ge 0$, then for every
model $m$ of $\delta$,
%$(m(\vec{v_0}),m(\vec{v})) \models \delta(\vec{v_0},\vec{v})$
there is an execution of $P$ that begins in $m(\vec{v}_0)$ and ends in $m(\vec{v})$,
using a call-stack bounded by $n$. Similarly, if $\delta(\vec{v}_0,\vec{v}) \in
\sigma(P, n)$, then $P(\vec{v}) \models_n \delta(\vec{v}_0,\vec{v})$.

%The \emph{reachability} map $\rho$ maps each procedure $P(\vec{v})
%\in \Prg$ and a number $n \ge 0$ to a set of assertions over $\vec{v_0} \cup
%\vec{v}$ that under-approximate the behavior of $P$ for a call-stack bounded by
%$n$. That is, if $\delta(\vec{v_0},\vec{v}) \in \rho(P, n)$ then for every
%model $(m(\vec{v_0}),m(\vec{v})) \models \delta(\vec{v_0},\vec{v})$ there is an
%execution of $P$ that starts in $m(\vec{v_0})$ and terminates in $m(\vec{v})$,
%using a call-stack bounded by $n$. Similarly, the \emph{summary} map $\sigma$ maps each
%procedure $P \in \Prg$ and a number $n \ge 0$ to a set of assertions over $\vec{v_0} \cup
%\vec{v}$ that over-approximate the behaviors of $P$ for a call-stack bounded by
%$n$. That is, for every execution of $P$ starting in a state $s(\vec{v})$ and
%terminating in $t(\vec{v})$ using a call-stack bounded by $n$,
%$(s(\vec{v_0}),t(\vec{v})) \models \delta(\vec{v_0},\vec{v})$ for all
%$\delta(\vec{v_0},\vec{v}) \in \sigma(P,n)$.

$\recmc$ alternates between two steps: (\textbf{A}) deciding
bounded safety (that also updates $\rho$ and $\sigma$ maps) and (\textbf{B})
checking whether the current proof of bounded safety is inductive
(\ie independent of the bound). It terminates when a counterexample or a
proof is found.

% To check safety in a property-driven fashion, \recmc maintains two
% maps of approximations for the semantics of the procedures in
% $\mathcal{P}$ for every value of $n$. The map of \emph{reachability
%   facts}, $\rho$, corresponds to what is known to be concretely
% reachable in a procedure and denotes an under-approximation of the
% semantics. If reachability facts of $M$ show that there is an
% execution consistent with $\neg\psi_M$, the property fails to hold. On
% the other hand, the map of \emph{summary facts}, $\sigma$, summarizes
% what cannot be reachable in a procedure and denotes an
% over-approximation of the semantics. If the summary facts at a given
% $n$ are inductive, $\sigma$ constitutes an inductive safety proof of
% the property.  The maps are initialized to $\bot$ and $\top$,
% respectively.

\begin{figure}[t]
\centering
\includegraphics[scale=.9]{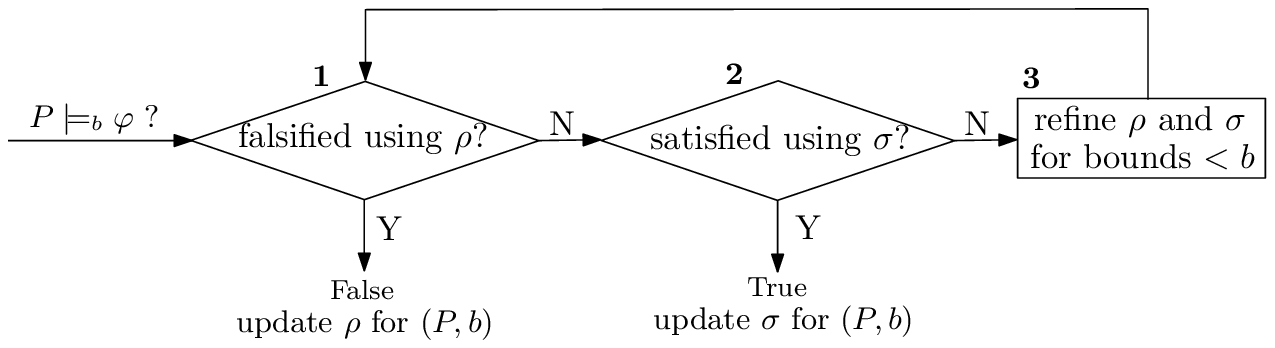}
\caption{Flow of the algorithm \bdsafety to check $P \models_b \varphi$.}
\label{fig:bdsafety}
\end{figure}

Bounded safety, $P \models_b \varphi$, is decided using $\bndsafety$ shown in
Fig.~\ref{fig:bdsafety}.  Step~{\bf 1} checks
whether $\varphi$ is falsified by current reachability facts in $\rho$ of
the callees of $P$. If so, it infers a new reachability fact
for $P$ at bound $b$ witnessing the falsification of $\varphi$. Step~{\bf 2}
checks whether $\varphi$ is satisfied using current summary
facts in $\sigma$ of the callees. If so, it infers a new summary fact for
$P$ at bound $b$ witnessing the satisfaction of $\varphi$. If
the prior two steps fail, there is a potential counterexample $\pi$ in
$P$ with a call to some procedure $R$ such that the
reachability facts of $R$ are too strong to witness $\pi$, but the
summary facts of $R$ are too weak to block it. Step~{\bf 3} updates $\rho$ and
$\sigma$ by creating (and recursively deciding) a new bounded
safety problem for $R$ at bound $b-1$.

We conclude this section with an illustration of $\recmc$ on the program
in Fig.~\ref{fig:overview_eg_prog} (adapted from~\cite{ed_paper}).
The program has 3 procedures: the main procedure \texttt{M}, and
procedures \texttt{T} and \texttt{D}. \texttt{M} calls \texttt{T} and
\texttt{D}. \texttt{T} modifies its argument \texttt{t} and calls
itself recursively. \texttt{D} decrements its argument \texttt{d}. Let the
property be $\varphi = m_0 \geq 2m + 4$.

\begin{figure}[t]
\centering
\begin {subfigure}{.3\columnwidth}
\centering
\begin{myverbbox}[\scriptsize]{\mainproc}
M (m) {
    T (m);
    D (m);
    D (m); }
\end{myverbbox}
%$\{m=m_0\}$
\mainproc\\
%$\{m_0 \ge 2m + 4\}$
\end{subfigure}
\begin {subfigure}{.3\columnwidth}
\centering
\begin{myverbbox}[\scriptsize]{\tandemproc}
T (t) {
    if (t>0) {
        t := t-2;
        T (t);
        t := t+1; } }
\end{myverbbox}
\tandemproc
\end{subfigure}
\begin {subfigure}{.3\columnwidth}
\centering
\begin{myverbbox}[\scriptsize]{\decrproc}
D (d) {
    d := d-1;
}
\end{myverbbox}
\decrproc
\end{subfigure}
\caption {A recursive program with 3 procedures.}
\label {fig:overview_eg_prog}
\end{figure}

\begin{figure}[t]
\centering
\includegraphics[scale=1]{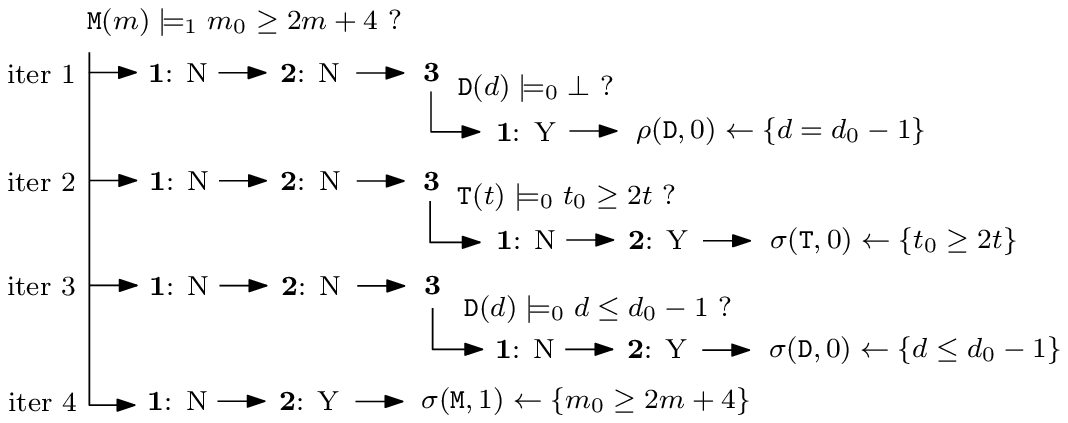}
\caption{A run of \bdsafety on program in
Fig.~\ref{fig:overview_eg_prog} and a bound 1 on the stack depth. Numbers in
bold refer to the steps in Fig.~\ref{fig:bdsafety}.}
\label{fig:overview_eg}
\end{figure}

The first iteration of $\recmc$ is trivial. The bound $n=0$ and since
\texttt{M} has no call-free executions it vacuously satisfies any
bounded safety property. Fig.~\ref{fig:overview_eg} shows the four
iterations of \bdsafety for the second iteration of $\recmc$ where $n=1$. For
this bound, the maps $\rho$ and $\sigma$ are initially empty. The first
iteration of \bdsafety finds a potential counterexample path in \texttt{M} and
the approximation for \texttt{D} is updated with a new reachability fact: $d = d_0 -1$. In the second iteration,
the approximation for \texttt{T} is updated. Note that the two calls to
\texttt{D} are ``jumped over'' using the reachability fact for
\texttt{D} computed in the first iteration. The new summary fact for \texttt{T} is: $t_0 \geq
2t$. In the third iteration, the approximation for \texttt{D} is updated again, now with a
summary fact $d \leq d_0 -1$. Finally, the summary facts for
\texttt{T} and \texttt{D} at bound $0$ are sufficient to establish bounded safety
at $n=1$. At this point, the summary map $\sigma$ is:
\begin{align*}
  \sigma(\texttt{M}, 1) &= \{m_0 \geq 2m + 4\} &
  \sigma(\texttt{T}, 0) &= \{t_0 \geq 2t\} &
  \sigma(\texttt{D}, 0) &= \{d \leq d_0 -1\}
\end{align*}
Ignoring the bounds, $\sigma$ is inductive. For example, we can prove that the
body of \texttt{T} satisfies $t_0 \geq 2t$, assuming that the calls do. Thus,
step \textbf{B} of $\recmc$ succeeds and the algorithm terminates declaring the program
SAFE.  In the rest of the paper, we show how to automate $\recmc$ using an
SMT-oracle.
%for the theory of program statements and assertions.
%Our key insight in an effective method for generalizing
%reachability facts and refinement queries from concrete executions. \ak{this
%para needs work.}

%%% Local Variables:
%%% mode: latex
%%% TeX-master: "main"
%%% End:

%% file: prelims.tex
\section {Preliminaries}
\label{sec:prelims}

\newcommand{\vale}[2]{\mathcal{V}_{#1}\left({#2}\right)}

Consider a first-order language with equality and let $\sig$ be its signature, \ie the set of
non-logical function and predicate symbols (including equality). An $\sig$-\emph{structure} $I$ consists of a domain of
interpretation, denoted $|I|$, and assigns elements of $|I|$ to variables, and
functions and predicates on $|I|$ to the symbols of $\sig$. Let $\varphi$ be
a formula. We assume the usual definition of satisfaction of $\varphi$ by
$I$, denoted $I \models \varphi$. $I$ is called a \emph{model} of $\varphi$ iff
$I \models \varphi$ and this can be extended to a set of formulas. A
first-order $\sig$-\emph{theory} $\thry$ is a set of deductively closed
$\sig$-sentences.  $I$ satisfies $\varphi$ modulo $\thry$, denoted $I
\models_\thry \varphi$, iff $I \models \thry \cup \{\varphi\}$.  $\varphi$ is
\emph{valid} modulo $\thry$, denoted $\models_\thry \varphi$, iff every model of
$\thry$ is also a model of $\varphi$.

Let $I$ be an $\sig$-structure and $\vec{w}$ be a list of fresh
function/predicate symbols not in $\sig$. A $(\sig \cup \vec{w})$-structure
$J$ is called an \emph{expansion} of $I$ to $\vec{w}$ iff $|J| = |I|$ and $J$
agrees with $I$ on the assignments to all variables and the symbols of $\sig$. We use the
notation $\expansion{I}{\vec{w}}{\vec{u}}$ to denote the expansion of $I$ to
$\vec{w}$ that assigns the function/predicate $u_i$ to the symbol $w_i$.
%
%We write $\sigma[w]$ for the \emph{extension}
%of $\sigma$ with a list $w$ of new function and predicate symbols. A
%$\sigma[w]$-structure $I'$ is called an \emph{expansion} of $I$ to $\sigma[w]$
%iff $|I|=|I'|$ and $I$ agrees with $I'$ on all variables and symbols except $w$.
%We write $I[u_1,\dots,u_{|w|}]$ for the expansion of $I$ to $\sigma[w]$ that
%assigns the interpretation $u_i$ to $w_i$.
For an $\sig$-sentence $\varphi$, we
write $I(\varphi)$ to denote the truth value of $\varphi$ under $I$. For a
formula $\varphi(\vec{x})$ with free variables $\vec{x}$, we overload the
notation $I(\varphi)$ to mean $\{\vec{a} \in |I|^{|\vec{x}|} \mid
\expansion{I}{\vec{x}}{\vec{a}} \models \varphi\}$. For simplicity of presentation, we sometimes identify the
truth value \emph{true} with $|I|$ and \emph{false} with $\emptyset$.

We assume that programs do not have internal procedures and that procedures
cannot be passed as parameters. Furthermore, without loss of generality, we
assume that programs do not have loops or global variables. In the following, we
define programs using a logical representation, as opposed to giving a
concrete syntax.
%
%Let $P(\vec{v})$ be a procedure with parameters $\vec{v}$ and let $\vec{v}_0$ be fresh variables not
%appearing in $P$. Clarke showed that every Hoare-triple
%$\htriple{\varphi}{P(\vec{v})}{\psi}$ can be obtained from a triple of the form
%$\htriple{\vec{v}=\vec{v}_0}{P(\vec{v})}{\delta(\vec{v}_0,\vec{v})}$ using a
%\emph{Rule of Adaptation}~\cite{ed_paper}. Intuitively, $\delta$
%over-approximates the semantics of $P$. 
%
%We assume that there are no nested procedures and no procedural
%parameters. Furthermore, without loss of generality, we assume that programs have no loops
%and no global variables. We do not define a
%concrete syntax of programs. Instead, we use a logic-based
%presentation, where the body of each procedure is described by a formula in
%first-order logic with equality over a given theory
%$\mathit{Th}$. There are many tools that translate programs in an
%imperative programming language like C to such a logic-based
%formulation (\eg \textsc{Boogie}~\cite{boogie},
%\textsc{Ufo}~\cite{ufo}).
%
% \subheading{Programs}
A \emph{program} $\Prg$ is a finite list of procedures with a designated
\emph{main} procedure $M$ where the program begins. A \emph{procedure} $P$ is a
tuple $\langle \iparams{P}, \oparams{P}, \sum{P}, \locals{P}, \body{P} \rangle$,
where
%\begin{enumerate}
    (a) $\iparams{P}$ is the finite list of variables denoting the input
    values of the parameters,
    (b) $\oparams{P}$ is the finite list of variables denoting the output
    values of the parameters,
    (c) $\sum{P}$ is a fresh predicate symbol of arity $|\iparams{P}| +
    |\oparams{P}|$,
    (d) $\locals{P}$ is the finite list of local variables, and
    (e) $\body{P}$ is a quantifier-free sentence over the signature
        $(\sig \cup \{\sum{Q} \mid Q\in\Prg\} \cup \iparams{P} \cup \oparams{P} \cup
        \locals{P})$ in which a predicate symbol $\sum{Q}$ appears only positively.
We use $\formals{P}$ to denote $\iparams{P} \cup \oparams{P}$.

Intuitively, for a
procedure $P$, $\sum{P}$ is used to denote its semantics and $\body{P}$ encodes
its body using the predicate symbol $\sum{Q}$ for a call to the procedure $Q$.
We require that a predicate symbol $\sum{Q}$ appears only positively in
$\body{P}$ to ensure a fixed-point characterization of the semantics as shown
later on. For example, for the signature $\sig = \langle 0, \mathit{Succ}, -, +, \le, >, =
\rangle$, the program in Fig.~\ref{fig:overview_eg_prog} is
represented as  $\langle M, T, D \rangle$ with $M = \langle m_0, m, \sum{M},
\langle \ell_0,\ell_1 \rangle, \body{M} \rangle$, $T = \langle t_0, t, \sum{T},
\langle \ell_0,\ell_1 \rangle, \body{T} \rangle$ and $D = \langle d_0, d,
\sum{D}, \emptyset, \body{D} \rangle$, where
\begin{equation}
\label{eqn:oveview_eg}
\begin{aligned}
    &\body{M} = \sum{T}(m_0,\ell_0) \land \sum{D}(\ell_0,\ell_1)
    \land \sum{D}(\ell_1,m)
    \qquad\body{D} =  (d = d_0 - 1)
 \\
    &\body{T} = \left( t_0 \le 0 \land t_0 = t \right) ~\lor~ \left( t_0 > 0 \land
    \ell_0 = t_0 - 2 \land \sum{T}(\ell_0,\ell_1) \land t = \ell_1 + 1 \right)
\end{aligned}
\end{equation}

Here, we abbreviate $\mathit{Succ}^i(0)$ by $i$ and $(m_0, t_0, d_0)$ and
$(m,t,d)$ denote the input and the output values of the parameters of the
original program, respectively. For a procedure $P$, let $\paths{P}$ denote the
set of all prime-implicants of $\body{P}$. Intuitively, each element of
$\paths{P}$ encodes a path in the procedure.

Let $\Prg = \langle P_0,\dots,P_n \rangle$ be a program and $I$ be an
$\sig$-structure. Let $\vec{X}$ be a list of length $n$ such that each $X_i$
is either (i) a truth value if $|\formals{P_i}| = 0$, or (ii) a subset of
$|I|^{|\formals{P_i}|}$ if $|\formals{P_i}| \ge 1$.  Let $J(I,\vec{X})$ denote
the expansion $\expansion{\expansion{I}{\sum{P_0}}{X_0}\dots}{\sum{P_n}}{X_n}$.
The \emph{semantics} of a procedure $P_i$ given $I$, denoted $\sem{P_i}_I$,
characterizes all the terminating executions of $P_i$ and is defined as follows.
$\langle \sem{P_0}_I, \dots, \sem{P_n}_I \rangle$ is the (pointwise) least
$\vec{X}$ such that for all $Q \in \Prg$, $J(I,\vec{X}) \models \forall
\formals{Q} \cup \locals{Q} \st (\body{Q} \implies \sum{Q}(\formals{Q}))$. This
has a well-known least fixed-point characterization~\cite{ed_paper}.

For a bound $b \ge 0$ on the call-stack, the \emph{bounded semantics} of a
procedure $P_i$ given $I$, denoted $\sem{P_i}^b_I$, characterizes all the
executions using a stack of depth bounded by $b$ and is defined by induction on
$b$:
\begin{align*}
    \sem{P_i}^0_I = J(I,\langle \emptyset,\dots,\emptyset \rangle)
                        (\exists \locals{P_i} \st \body{P_i}),
&\quad
    \sem{P_i}^b_I = J(I,\langle \sem{P_0}^{b-1}_I, \dots, \sem{P_n}^{b-1}_I \rangle)
                        (\exists \locals{P_i} \st \body{P_i})
\end{align*}
%We write $\sem{P}^b_I$ to denote $\sem{\Prg}^b_I(\sum{P})$.

An \emph{environment} is a function that maps a predicate symbol $\sum{P}$ to a
formula over $\formals{P}$. Given a formula $\tau$ and an environment $E$,
we abuse the notation $\sem{\cdot}$ and write $\sem{\tau}_E$ for the formula
obtained by instantiating every predicate symbol $\sum{P}$ by $E(\sum{P})$ in $\tau$.
Let $\thry$ be an $\sig$-theory. A \emph{safety property} for a procedure $P
\in \Prg$ is a formula over
$\formals{P}$. $P$ satisfies a safety property $\varphi$ w.r.t $\thry$,
denoted $P \models_\thry \varphi$, iff for all models $I$ of $\thry$, $\sem{P}_I
\subseteq I(\varphi)$.
A \emph{safety property} $\psi$ of the program $\Prg$ is a safety property of
its main procedure.  A \emph{safety proof} for $\psi(\formals{M})$ is an
environment $\Pi$ that is both safe and inductive:
\begin{align*}
    \models_\thry \sem{\forall \vec{x} \st \sum{M}(\vec{x}) \implies
    \psi(\vec{x})}_\Pi,
&\quad
    \forall P \in \Prg \st \models_\thry \sem{\forall \formals{P} \cup \locals{P} \st
    (\body{P} \implies \sum{P}(\formals{P}))}_\Pi
\end{align*}

Given a formula $\varphi(\formals{P})$ and $b \ge 0$, a procedure $P$
satisfies \emph{bounded safety} w.r.t $\thry$, denoted $P
\models_{b,\thry} \varphi$, iff for all
models $I$ of $\thry$, $\sem{P}^b_I \subseteq I(\varphi)$. In this case, we
also call $\varphi$ a \emph{summary fact} for $\langle P,b \rangle$. We call
$\varphi$ a \emph{reachability fact} for $\langle P,b \rangle$ iff $I(\varphi)
\subseteq \sem{P}^b_I$, for all models $I$ of $\thry$.  Intuitively,
\emph{summary facts} and \emph{reachability facts} for $\langle P,b \rangle$,
respectively, over- and under-approximate $\sem{P}^b_I$ for every model $I$ of
$\thry$.

%Let $\varphi(\params{P})$ be an assertion and $b\ge 0$ be a natural
%number.  The problem of \emph{bounded safety} for $b$ is to determine
%whether $\sem{P}^b \subseteq \{\vec{a} \mid
%\vale{}{\varphi[\params{P} \gets \vec{a}]}\}$. We write $P \models_b
%\varphi$ to denote that bounded safety holds for $P$, $\varphi$ and $b$. In
%this case, we say that $\varphi$ is a \emph{summary fact} for $\langle
%P,b \rangle$. We call assertion $\varphi$ a \emph{reachability fact} for
%$\langle P,b \rangle$ iff $\{\vec{a} \mid \vale{}{\varphi[\params{P}
  %\gets \vec{a}]}\} \subseteq \sem{P}^b$.  Intuitively, summary facts
%and reachability facts, respectively, over- and under-approximate the
%bounded semantics $\sem{P}^b$.

A \emph{bounded assertion map} maps a procedure $P$ and a natural number $b \ge
0$ to a set of formulas over $\formals{P}$. Given a
bounded assertion map $m$ and $b \ge 0$, we define two special environments
$U_m^b$ and $O_m^b$ as follows.
\begin{align*}
    %U_m^b : \sum{P} \mapsto \bigor_{\delta \in m(P,b)} \delta &
    U_m^b : \sum{P} \mapsto \bigor \{\delta \in m(P,b') \mid b' \le b\} &
    \quad\quad
    %O_m^b : \sum{P} \mapsto \bigand_{\delta \in m(P,b)} \delta
    O_m^b : \sum{P} \mapsto \bigand \{\delta \in m(P,b') \mid b' \ge b\}
\end{align*}
% Intuitively, if $m$ maps to reachability facts, $U_m^b$
% \emph{under-approximates} the bounded semantics. Similarly, if $m$ maps to
% summary facts, $O_m^b$ over-approximates the bounded semantics.
We use $U_m^b$ and $O_m^b$ to under- and over-approximate the bounded semantics.
For convenience, let $U_m^{-1}$ and $O_m^{-1}$ be environments that map every
symbol to $\bot$.

%% file: algo_arxiv.tex
\newcommand{\node}[3]{\langle {#1}, {#2}, {#3} \rangle}
\newcommand{\nodes}{\mathcal{Q}}

% a state is a triple <nodes, reach facts, sum facts>
\newcommand{\state}[3]{{#1} \parallel {#2} \parallel {#3}}

% mathpartir weirdities
\newcommand{\equals}{=}
\newcommand{\comma}{,}
\newcommand{\leftsq}{[}
\newcommand{\rightsq}{]}

% auxiliary
\newcommand{\subst}[2]{{#1} \leftarrow {#2}}
\newcommand{\pre}[2]{\text{\tt{pre}}[#1]({#2})}

\section{Model Checking Recursive Programs}
\label{sec:algo}

% motivate the incremental approach
%   problem is semidecidable, we want to find (the shortest) cex if it exists
%   essence of model checking -- least fixed point computation, see what happens
%       for n iterations and compute what happens for n+1 iterations
%   [iterative deepening is the best heuristic choice in an unknown world -- Russel and Norvig?]

% what is the goal?
%   given a reachability query, find sufficient set of facts about procedures to
%       answer the query -- output either a cex or a Hoare-style proof
%   2 kinds of summary facts
%       phi => P (under),
%       P => phi (over)

%       Say somewhere that the premises of the rules need an oracle
%       for satisfiability

In this section, we present our algorithm  $\recmc (\Prg,
\varphi_{\mathit{safe}})$ that determines whether a program $\Prg$ satisfies a
safety property $\varphi_{\mathit{safe}}$. Let $\sig$ be the signature of the
first-order language under consideration and assume a fixed $\sig$-theory $\thry$.
To avoid clutter, we drop the subscript $\thry$ from
the notation $\models_\thry$ and $\models_{b,\thry}$. We also establish the soundness and complexity of \recmc. An
efficient instantiation of $\recmc$ to Linear Arithmetic is presented in
Section~\ref{sec:lazy_qe}.

\begin{figure}[t]
\begin{subfigure}{.51\textwidth}
\begin{algorithm}[H]
\scriptsize
\DontPrintSemicolon
\textsc{RecMC}$(\Prg, \varphi_{\mathit{safe}})$\;
\nl $n \gets 0 \mathbin{;} \rho\gets \emptyset \mathbin{;} \sigma \gets \emptyset$ \;
\nl \While {true} {
\nl     $\mathit{res},\rho,\sigma \gets \bndsafety(\Prg, \varphi_{\mathit{safe}}, n,\rho,\sigma)$\;
\nl     \If {res \emph{is} UNSAFE} {
\nl         \Return UNSAFE, $\rho$
        } \Else {
\nl         $\mathit{ind}, \sigma \gets \textsc{CheckInductive}(\Prg,\sigma,n)$
\nl         \If {ind} {
\nl             \Return SAFE, $\sigma$
            }
\nl         $n \gets n + 1$ \;
        }
    }
\end{algorithm}
\end{subfigure}
\begin{subfigure}{.6\textwidth}
\begin{algorithm}[H]
\scriptsize
\DontPrintSemicolon
\setcounter{AlgoLine}{9}
\textsc{CheckInductive}($\Prg$, $\sigma$, $n$)\;
\nl $\mathit{ind} \gets \mathit{true}$\;
\nl \ForEach {$P \in \Prg$} {
\nl     \ForEach {$\delta \in \sigma(P,n)$} {
\nl         \If {$\models \sem{\body{P}}_\sigma^n \implies \delta$} {
\nl             $\sigma \gets \sigma \cup (\langle P, n+1 \rangle \mapsto \delta)$
            } \Else {
\nl             $\mathit{ind} \gets \mathit{false}$
            }
        }
    }
\nl \Return $(\text{\emph{ind}}, \sigma)$
\end{algorithm}
\end{subfigure}
\vspace{-0.2in}
\caption{Pseudo-code of \textsc{RecMC}.}
\label{fig:unbounded}
\end{figure}

\paragraph{\textbf{Main Loop}.} \recmc maintains two \emph{bounded
  assertion maps} $\rho$ and $\sigma$ for reachability and summary facts,
respectively. For brevity, for a first-order formula $\tau$, we write
$\sem{\tau}_\rho^b$ and $\sem{\tau}_\sigma^b$ to denote
$\sem{\tau}_{U_\rho^b}$ and $\sem{\tau}_{O_\sigma^b}$,
respectively, where the environments $U_\rho^b$ and $O_\sigma^b$ are
as defined in Section~\ref{sec:prelims}. Intuitively,
$\sem{\tau}_\rho^b$ and $\sem{\tau}_\sigma^b$, respectively, under-
and over-approximate $\tau$ using $\rho$ and $\sigma$.

The pseudo-code of the main loop of \recmc (corresponding to the flow
diagram in Fig.~\ref{fig:recmc}) is shown in Fig.~\ref{fig:unbounded}.
\recmc follows an \emph{iterative deepening} strategy. In each
iteration, $\bndsafety$ (described below) checks whether all
executions of $\Prg$ satisfy $\varphi_\mathit{safe}$ for a bound $n
\ge 0$ on the call-stack, \ie if $M \models_n \varphi_{\mathit{safe}}$.
$\bndsafety$ also updates the maps $\rho$ and $\sigma$. Whenever $\bndsafety$ returns
$\mathit{UNSAFE}$, the reachability facts in $\rho$ are sufficient to
construct a counterexample and the loop terminates. Whenever
$\bndsafety$ returns $\mathit{SAFE}$, the summary facts in $\sigma$
are sufficient to prove the absence of a counterexample for the
current bound $n$ on the call-stack. In this case, if $\sigma$ is also
inductive, as determined by \textsc{CheckInductive}, $O_\sigma^n$ is a
safety proof and the loop terminates.  Otherwise, the bound on the
call-stack is incremented and a new iteration of the loop begins. Note that, as
a side-effect of \textsc{CheckInductive}, some
summary facts are propagated to the bound $n+1$. This is similar to
\emph{push generalization} in IC3~\cite{ic3}.

\begin{figure}[t]
\begin{mathpar}
    \inferrule*[left=Init]
        { }
        { \state {\{\node {M} {\neg\varphi_\text{\emph{safe}}} {n}\}}
                 {\rho_{\mathit{Init}}}
                 {\sigma_{\mathit{Init}}}
        }\\
    \inferrule*[left=Sum]
        { \state {\nodes} {\rho} {\sigma} \\
          \node {P} {\varphi} {b} \in \nodes \\
          \models \sem{\body{P}}_\sigma^{b-1} \implies \neg\varphi }
        { \state {\nodes \setminus \{\node {P} {\eta} {c} ~|~ c \le b, \models \sem{\sum{P}}_\sigma^c \land \psi \implies \neg\eta\}}
                 {\rho}
                 {\sigma \union \{\langle P, b \rangle \mapsto \psi\}} }
\end{mathpar}
        \flushright{$\textnormal{where}~ \psi \equals
            \textsc{Itp} (\sem{\body{P}}_\sigma^{b-1} \comma \neg\varphi)$}
\begin{mathpar}
    \inferrule*[left=Reach]
        { \state {\nodes} {\rho} {\sigma} \\
          \node {P} {\varphi} {b} \in \nodes \\
          \pi \in \paths{P} \\
          \not\models \sem{\pi}_\rho^{b-1} \implies \neg\varphi }
        { \state {\nodes \setminus \{\node {P} {\eta} {c} ~|~ c\ge b, \not\models \psi \implies \neg\eta\}}
                 {\rho \union \{\langle P, b \rangle \mapsto \psi\}}
                 {\sigma} }
\end{mathpar}
        \flushright{$\textnormal{where}~ \psi \equals
            \exists \locals{P} \st\sem{\pi}_\rho^{b-1}$}
\begin{mathpar}
    \inferrule*[left=Query]
        { \state {\nodes} {\rho} {\sigma} \\
          \node {P} {\varphi} {b} \in \nodes \\
          \models \sem{\body{P}}_\rho^{b-1} \implies \neg\varphi \\
          \pi \in \paths{P} \\\\
          \pi = \pi_u \land \sum{R}(\vec{a}) \land \pi_v \\
          \models \sem{\pi_u}_\sigma^{b-1} \land \sem{\sum{R}(\vec{a})}_\rho^{b-1} \land \sem{\pi_v}_\rho^{b-1} \implies \neg\varphi \\
          \not\models \sem{\pi_u}_\sigma^{b-1} \land \sem{\sum{R}(\vec{a})}_\sigma^{b-1} \land \sem{\pi_v}_\rho^{b-1} \implies \neg\varphi }
        { \state {\nodes \union \{\node {R} {\psi} {b-1}\}}
                 {\rho} {\sigma} }
\end{mathpar}
        \flushright{$\textnormal{where}~
                    \begin{cases}
                        \psi \equals
                        \left( \exists \left( \formals{P} \cup \locals{P} \right) \setminus \vec{a} \st
                                    \sem{\pi_u}_\sigma^{b-1} \land
                                    \sem{\pi_v}_\rho^{b-1} \land \varphi \right)
                              [\vec{a} \leftarrow \formals{R}]\\
                        \textnormal{for all } \node{R}{\eta}{b-1} \in \nodes, \models \psi \implies \neg\eta
                    \end{cases}$}
\begin{mathpar}
    \inferrule*[left=Unsafe]
        { \state {\emptyset}
                 {\rho}
                 {\sigma} \\
          \not\models \sem{\sum{M}}_\rho^n \implies \varphi_\text{\emph{safe}} }
        { \text{\emph{UNSAFE}} }

    \inferrule*[left=Safe]
        { \state {\emptyset}
                 {\rho}
                 {\sigma} \\
          \models \sem{\sum{M}}_\sigma^n \implies \varphi_\text{\emph{safe}} }
        { \text{\emph{SAFE}} }
\end{mathpar}
\caption {Rules defining $\bndsafety(\Prg,\varphi_{\mathit{safe}},n,\rho_{\mathit{Init}},\sigma_{\mathit{Init}})$.}
\label{fig:basic_algo}
\end{figure}

\paragraph{\textbf{Bounded safety}.} We describe the routine
$\bndsafety(\Prg,\varphi_{\mathit{safe}},n,\rho_{\mathit{Init}},\sigma_{\mathit{Init}})$ as an abstract
transition system~\cite{DBLP:journals/jacm/NieuwenhuisOT06} defined by the inference rules
shown in Fig.~\ref{fig:basic_algo}. Here, $n$ is the current bound on the
call-stack and $\rho_{\mathit{Init}}$ and $\sigma_{\mathit{Init}}$ are the maps of reachability and summary
facts input to the routine. A state of $\bndsafety$ is a triple $\state
{\nodes}{\rho}{\sigma}$, where $\rho$ and $\sigma$ are the current maps and
$\nodes$ is a set of triples $\langle P, \varphi, b \rangle$ for a procedure
$P$, a formula $\varphi$ over $\formals{P}$, and a number $b \ge 0$.
A triple $\langle P, \varphi, b \rangle \in \nodes$ is called a \emph{bounded
reachability query} and asks whether $P \not\models_b \neg\varphi$, \ie whether
there is an execution in $P$ using a call-stack bounded by $b$ where the values
of $\formals{P}$ satisfy $\varphi$.

\bndsafety starts with a single query $\langle M, \neg
\varphi_\mathit{safe}, n \rangle$ and initializes the maps of
reachability and summary facts (rule \textsc{Init}). It checks whether $M
\models_n \varphi_\mathit{safe}$ by inferring new summary and reachability
facts to answer existing queries (rules \textsc{Sum} and \textsc{Reach}) and generating new queries
(rule \textsc{Query}). When there are no queries left to answer, \ie $\nodes$ is
empty, it terminates with a result of either $\mathit{UNSAFE}$ or $\mathit{SAFE}$ (rules
\textsc{Unsafe} and \textsc{Safe}).

\subheadingnodot{\textsc{Sum}} infers a new summary fact when a query
$\langle P, \varphi, b \rangle$ can be answered negatively. In this
case, there is an over-approximation of the bounded semantics of $P$ at $b$, obtained using
the summary facts of callees at bound $b-1$, that is unsatisfiable with
$\varphi$. That is, $\models \sem{\body{P}}_\sigma^{b-1} \implies
\neg\varphi$. The inference of the new fact is by
interpolation~\cite{craig} (denoted by \textsc{Itp} in the
side-condition of the rule). Thus, the new summary fact $\psi$ is a
formula over $\formals{P}$ such that $\models \left( \sem{\body{P}}_\sigma^{b-1}
\implies \psi(\formals{P}) \right) ~\land~ (\psi(\formals{P}) \implies
\neg\varphi)$.
Note that $\psi$ over-approximates the bounded semantics of $P$ at $b$.  Every
query $\langle P,\eta,c \rangle \in \nodes$ such that $\eta$ is unsatisfiable
with the updated environment $O^c_\sigma(\sum{P})$ is immediately answered and
removed.

\subheadingnodot{\textsc{Reach}} infers a new reachability fact when a query
$\langle P, \varphi, b \rangle$ can be answered positively. In this case,
there is an under-approximation of the bounded semantics of $P$ at $b$, obtained
using the reachability facts of callees at bound $b-1$, that is satisfiable
with $\varphi$. That is, $\not\models \sem{\body{P}}_\rho^{b-1}
\implies \neg\varphi$. In particular, there exists a path $\pi$ in
$\mathit{Paths}(P)$ such that $\not\models \sem{\pi}_\rho^{b-1} \implies
\neg\varphi$. The new reachability fact $\psi$ is obtained by choosing such a
$\pi$ non-deterministically and existentially quantifying all local variables
from $\sem{\pi}_\rho^{b-1}$. Note that $\psi$ under-approximates the bounded
semantics of $P$ at $b$.  Every query $\langle P,\eta,c \rangle \in \nodes$ such
that $\eta$ is satisfiable with the updated environment $U^c_\rho(\sum{P})$ is
immediately answered and removed.

\subheadingnodot{\textsc{Query}} creates a new query when a query $\langle
P,\varphi,b \rangle$ cannot be answered using current $\rho$ and $\sigma$. In
this case, the current over-approximation of the bounded semantics of $P$ at $b$
is satisfiable with $\varphi$ while its current under-approximation is
unsatisfiable with
$\varphi$. That is, $\not\models \sem{\body{P}}_\sigma^{b-1} \implies \neg\varphi$ and
$\models \sem{\body{P}}_\rho^{b-1} \implies \neg\varphi$. In particular, there exists
a path $\pi$ in $\mathit{Paths}(P)$ such that $\not\models \sem{\pi}_\sigma^{b-1}
\implies \neg\varphi$ and $\models \sem{\pi}_\rho^{b-1} \implies \neg\varphi$.
Intuitively, $\pi$ is a potential counterexample path that needs to be checked for
feasibility. Such a $\pi$ is chosen non-deterministically. $\pi$ is guaranteed
to have a call $\sum{R}(\vec{a})$ to a procedure $R$ such that the
under-approximation $\sem{\sum{R}(\vec{a})}_\rho^{b-1}$ is too strong to witness $\pi$ but the
over-approximation $\sem{\sum{R}(\vec{a})}_\sigma^{b-1}$ is too weak to block it. That is, $\pi$ can be
partitioned into a prefix $\pi_u$, a call $\sum{R}(\vec{a})$ to $R$, and a
suffix $\pi_v$ such that the following hold:
\begin{align*}
  \models \sem{\sum{R}(\vec{a})}_\rho^{b-1} &\implies \left((\sem{\pi_u}_\sigma^{b-1}
  \land \sem{\pi_v}_\rho^{b-1}) \implies \neg \varphi\right) \\
  \not\models \sem{\sum{R}(\vec{a})}_\sigma^{b-1} &\implies \left((\sem{\pi_u}_\sigma^{b-1}
  \land \sem{\pi_v}_\rho^{b-1}) \implies \neg \varphi\right)
\end{align*}
Note that the prefix $\pi_u$ and the suffix $\pi_v$ are over- and
under-approximated, respectively.
A new query $\langle R,\psi,b-1 \rangle$ is created where $\psi$ is obtained by
existentially quantifying all variables from $\sem{\pi_u}_\sigma^{b-1} \land
\sem{\pi_v}_\rho^{b-1} \land \varphi$ except the arguments $\vec{a}$ of the
call, and renaming appropriately. If the new query is answered
negatively (using \textsc{Sum}), all executions along $\pi$ where the values of
$\formals{P} \cup \locals{P}$ satisfy $\sem{\pi_v}_\rho^{b-1}$ are spurious counterexamples.
An additional side-condition requires that $\psi$ ``does not overlap''
with $\eta$ for any other query $\langle R,\eta,b-1 \rangle$ in $\nodes$.
This is necessary for termination of \bndsafety (Theorem~\ref{thm:termination}). In practice, the side-condition is
trivially satisfied by always applying the rule to $\langle P, \varphi, b
\rangle$ with the smallest $b$.

\begin{figure}[t]
\centering
\footnotesize
\begin{tabular}{c|c|c|c}
& $\pi_i$ & $\sem{\pi_i}_\rho^0$ & $\sem{\pi_i}_\sigma^0$ \\
\hline
$i=1$ & $\sum{T}(m_0,\ell_0)$ & $\bot$ & $\top$ \\
$i=2$ & $\sum{D}(\ell_0,\ell_1)$ & $\ell_1=\ell_0-1$ & $\top$ \\
$i=3$ & $\sum{D}(\ell_1,m)$ & $m=\ell_1-1$ & $\top$ \\
\end{tabular}
\caption {Approximations of the only path $\pi$ of the procedure $M$ in
Fig.~\ref{fig:overview_eg_prog}.}
\label{fig:query_eg}
\end{figure}

For example, consider the program in Fig.~\ref{fig:overview_eg_prog} represented
by (\ref{eqn:oveview_eg}) and the query $\langle M, \varphi, 1 \rangle$ where
$\varphi \equiv m_0 < 2m+4$.  Let $\sigma = \emptyset$, $\rho(D,0) =
\{d=d_0-1\}$ and $\rho(T,0) = \emptyset$. Let $\pi = (\sum{T}(m_0,\ell_0) \land
\sum{D}(\ell_0,\ell_1) \land \sum{D}(\ell_1,m))$ denote the only path in the
procedure $M$. Fig.~\ref{fig:query_eg} shows $\sem{\pi_i}_\rho^0$ and
$\sem{\pi_i}_\sigma^0$ for each conjunct $\pi_i$ of $\pi$. As the figure shows,
$\sem{\pi}_\sigma^0$ is satisfiable with $\varphi$, witnessed by the execution $e
\equiv \langle m_0=3, \ell_0=3, \ell_1=2, m=1 \rangle$. Note that this execution
also satisfies $\sem{\pi_2 \land \pi_3}_\rho^0$. But,
$\sem{\pi_1}_\rho^0$ is too strong to witness it, where $\pi_1$ is the call
$\sum{T}(m_0,\ell_0)$.  To create a new query for $T$, we first existentially quantify
all variables other than the arguments $m_0$ and $\ell_0$ from $\pi_2 \land
\pi_3 \land \varphi$, obtaining $m_0 < 2\ell_0$.  Renaming the arguments by the
parameters of $T$ results in the new query $\langle T,t_0<2t,0 \rangle$. Further
iterations of \bndsafety would answer this query negatively making the execution
$e$ spurious. Note that this would also make all other executions where the
values to $\langle m_0, \ell_0, \ell_1, m \rangle$ satisfy
$\sem{\pi_2\land\pi_3}_\rho^0$ spurious.

\paragraph{\textbf{Soundness and Complexity}.}
Soundness of $\recmc$ follows from that of
$\bndsafety$, which can be shown by a case analysis on the
inference rules\footnote{Proofs of all of the theorems are in the
  Appendix.}.

\begin{theorem}
\label{thm:soundness}
\bndsafety and \recmc are sound.
%In particular, 
%\begin{enumerate}
%\item If $\bndsafety(\Prg,\varphi,n,\emptyset,\emptyset)$ returns \emph{SAFE}
%$($\emph{UNSAFE}$)$, then $M \models_n \varphi$ $($respectively, $M \not\models_n
%\varphi$$)$.
%\item If $\recmc(\Prg,\varphi)$ returns \emph{SAFE} $($\emph{UNSAFE}$)$, then $M
%\models \varphi$ $($respectively, $M \not\models \varphi$$)$.
%\end{enumerate}
\end{theorem}

$\bndsafety$ is complete relative to an oracle for satisfiability modulo
$\thry$.  Even though the number of reachable states of a procedure is unbounded
in general, the number of reachability facts inferred by \bndsafety is finite.
This is because a reachability fact corresponds to a path (see \textsc{Reach})
and given a bound on the call-stack, the number of such facts is bounded. This
further bounds the number of queries that can be created.

\begin{theorem}
\label{thm:termination}
Given an oracle for $\thry$, $\bndsafety(\Prg, \varphi, n,
\emptyset, \emptyset)$ terminates.
%Let $\Prg$ be a program, $\varphi$ be a safety property, $p$ be the maximum
%number of paths in $P \in \Prg$, $c$ be the maximum number of procedure calls
%along any path in $\Prg$, and $N$ be the number of procedures in $\Prg$. Assume
%an oracle for satisfiability modulo $\thry$. Then, $\bndsafety(\Prg, \varphi, n,
%\emptyset, \emptyset)$ terminates in $O((N \cdot p)^{O(n^2)} \cdot n \cdot
%c^n)$-many applications of the rules in Fig.~\ref{fig:basic_algo}.
\end{theorem}

As a corollary of Theorem~\ref{thm:termination}, \recmc is a
co-semidecision procedure for safety, \ie \recmc is guaranteed to
find a counterexample if one exists. In contrast, the closest related algorithm
GPDR~\cite{gpdr} is not a co-semidecision procedure (see Appendix).
Finally, for Boolean Programs $\recmc$ is a complete decision
procedure. Unlike the general case, the number of reachable states of a Boolean
Program, and hence the number of reachability facts, is finite and independent
of the bound on the call-stack. Let $N = |\Prg|$ and $k = \max\{|\formals{P}| \mid P \in
\Prg\}$.

\begin{theorem}
\label{thm:bool_prog}
Let $\Prg$ be a Boolean Program. Then $\recmc(\Prg, \varphi)$ terminates in
$O(N^2 \cdot 2^{2k})$-many applications of the rules in
Fig.~\ref{fig:basic_algo}.
%Let $\Prg$ be a Boolean Program, $\varphi$ a safety property, $N$ the number of
%procedures in $\Prg$ and $a$ the maximum arity of $P \in \Prg$. Then,
%$\recmc(\Prg, \varphi)$ terminates in $O(N^2 \cdot 2^{2a})$-many applications
%of the rules in Fig.~\ref{fig:basic_algo}.
\end{theorem}

Note that due to the iterative deepening strategy of $\recmc$, the
complexity is quadratic in the number of procedures (and not linear as in~\cite{bebop}).
In contrast, other SMT-based algorithms, such as \textsc{Whale}~\cite{whale},
are worst-case exponential in the number of states of a Boolean Program.

In summary, \recmc checks safety of a recursive program by inferring
the necessary under- and over-approximations of procedure semantics
and using them to analyze procedures individually.
%This differs from the well-known algorithms for Boolean programs, \eg
%\textsc{RHS}~\cite{rhs}, \textsc{Bebop}~\cite{bebop}, which do not have a
%notion of an over-approximation of semantics. On the other hand, the algorithm
%GPDR~\cite{gpdr} uses interpolation to obtain over-approximations as in \recmc.
%However, GPDR does not have an equivalent of the \textsc{Reach} rule, \ie $\rho$
%is always $\emptyset$. Instead, it uses a modified \textsc{Query} rule which
%creates new queries based on a model satisfying the premises and caches
%previously answered queries. While the number of possible queries is finite for Boolean
%Programs, and caching ensures termination, the number of queries in GPDR is
%unbounded in general.

%%% Local Variables:
%%% mode: latex
%%% TeX-master: "main"
%%% End:

%% file: la_arxiv.tex
\newcommand{\qfml}{\psi}
\newcommand{\qffml}{\psi_f}
\renewcommand{\matrix}{\psi_m}

\section{Model Based Projection}
\label{sec:lazy_qe}
$\recmc$, as presented in Section~\ref{sec:algo}, can be used as-is,
when $\thry$ is Linear Arithmetic. But, note that the rules \textsc{Reach}
and \textsc{Query} introduce existential quantifiers in reachability
facts and queries. Unless eliminated, these quantifiers accumulate and
the size of the formulas grows exponentially in the bound on the
call-stack.
%the size of the assertions in the side-conditions of \textsc{Reach} and \textsc{Query} can grow
%exponentially in the bound of the call-stack because of the
%existential quantifiers.
Using quantifier elimination (QE) to eliminate the quantifiers is expensive.
Instead, we suggest an alternative
approach that under-approximates existential quantification with
quantifier-free formulas \emph{lazily} and efficiently. We first introduce
a model-based under-approximation of QE, which we call a \emph{Model Based Projection}
(MBP). Second, we give an efficient (linear in the size of formulas
involved) MBP procedure for Linear Rational Arithmetic (LRA). Due to
space limitations, MBP for Linear Integer Arithmetic (LIA) is described in the
Appendix. Finally, we show a modified version of $\bndsafety$
that uses MBP instead of existential quantification and show that it
is sound and terminating.

\paragraph{\textbf{Model Based Projection} $($MBP$)$.}
Let $\lambda(\vec{y})$ be the existentially quantified formula
$\exists \vec{x} \st \lambda_m(\vec{x},\vec{y})$ where $\lambda_m$ is quantifier free. A function $\projf{\lambda}$
from models (modulo $\thry$) of $\lambda_m$ to quantifier-free formulas over $\vec{y}$ is a \emph{Model Based
Projection} (for $\lambda$) iff it has a finite image, $\lambda \equiv \bigor_{M
\models \lambda_m} \projf{\lambda}(M)$, and for every model $M$ of $\lambda_m$, $M
\models \projf{\lambda}(M)$.
%\begin{align*}
  %\lambda \equiv \bigor_{M \models \lambda_m} \projf{\lambda}(M) &
  %\quad\quad
  %\forall M \models \lambda_m \st M \models \projf{\lambda}(M)
%\end{align*}

In other words, $\projf{\lambda}$ covers the space of all models of $\lambda_m(\vec{x},\vec{y})$ by a finite
set of quantifier-free formulas over $\vec{y}$. MBP exists for any theory
that admits quantifier elimination, because one can first obtain an equivalent
quantifier-free formula and map every model to it.

\paragraph{\textbf{MBP for Linear Rational Arithmetic}.} We begin with a brief
overview of Loos-Weispfenning (LW) method~\cite{lw} for quantifier elimination in LRA.
We borrow our presentation from Nipkow~\cite{nipkow} to which we refer the
reader for more details. Let $\lambda(\vec{y}) = \exists \vec{x} \st \lambda_m
(\vec{x},\vec{y})$ as above. Without loss of generality, assume that $\vec{x}$
is singleton, $\lambda_m$ is in Negation Normal Form, and $x$ only
appears in the literals of the form $\ell < x$, $x < u$, and $x=e$,
where $\ell$, $u$, and $e$ are $x$-free. Let $\lits{\lambda}$ denote
the literals of $\lambda$. The LW-method states that
\begin{equation}
  \label{eq:lw}
  \exists x \st \lambda_m(x) ~\equiv~ \left( \bigor_{(x=e) \in
    \lits{\lambda}} \lambda_m[e] ~\lor~ \bigor_{(\ell < x) \in
    \lits{\lambda}} \lambda_m [\ell + \epsilon] ~\lor~~ \lambda_m[-\infty] \right)
\end{equation}
where $\lambda_m[\cdot]$ denotes a \emph{virtual substitution} for the literals containing $x$.
Intuitively, $\lambda_m[e]$ covers the case when a literal $(x=e)$ is true.
Otherwise, the set of $\ell$'s in the literals $(\ell < x)$ identify intervals
in which $x$ can lie which are covered by the remaining substitutions.
%Note that $-\infty$ and $\epsilon$ in~\eqref{eq:lw} are special
%symbolic constants that are substituted into $\lambda_m$
%\emph{virtually}.
We omit the details of the substitution and instead
illustrate it on an example. Let $\lambda_m$ be $(x=e \land \phi_1) \lor
(\ell < x \land x < u) \lor (x<u \land \phi_2)$, where
$\ell,e,u,\phi_1,\phi_2$ are $x$-free. Then,
\begin{align*}
     \exists x \st \lambda_m(x) &\equiv \lambda_m[e] \lor \lambda_m[\ell + \epsilon] \lor \lambda_m[-\infty] \\
            &\equiv \big( \phi_1 \lor \left( \ell < e \land e < u \right) \lor
            \left( e<u \land \phi_2 \right) \big) \lor \big( \ell < u \lor
            (\ell<u \land \phi_2) \big) \lor \phi_2 \\
            &\equiv \phi_1 \lor (\ell < u) \lor \phi_2
\end{align*}

We now define an MBP $\lraprojf{\lambda}$ for LRA as a map from models
of $\lambda_m$ to disjuncts in~\eqref{eq:lw}. Given $M \models
\lambda_m$, $\lraprojf{\lambda}$ picks a disjunct that covers $M$
based on values of the literals of the form $x=e$ and $\ell < x$ in $M$. Ties are
broken by a syntactic ordering on terms (\eg when $M \models \ell'
= \ell$ for two \mbox{literals $\ell < x$ and $\ell' < x$).}
%A model $M$ of
%$\lambda_m(x)$ is mapped to $\lambda_m(e)$ if $M \models (x=e)$ for
%$(x=e) \in \lits{\lambda_m}$. $M$ is mapped to $\lambda_m(\ell +
%\epsilon)$ if $M \models (\ell < x)$ for $(\ell < x) \in
%\lits{\lambda_m}$ and $\ell$ is the largest such bound in
%$M$. Otherwise, $M$ is mapped to $\lambda_m (-\infty)$. Ties are
%broken by a syntactic ordering on terms (\ie when $M \models \ell' =
%\ell$ for two literals $\ell < x$ and $\ell' < x$). Formally,
\[
    \lraprojf{\lambda}(M) = \begin{cases}
                    \lambda_m[e], & \text{if } (x=e)\in \lits{\lambda} \land M \models x=e \\
                    \lambda_m[\ell+\epsilon], & \text{else if } (\ell < x) \in \lits{\lambda} \land M \models \ell<x \land{} \\
                                & \forall (\ell'<x) \in \lits{\lambda} \st M \models \left( (\ell' < x) \implies (\ell' \le \ell) \right) \\
                    \lambda_m[-\infty], & \text{otherwise}
                    \end{cases}
\]

\begin{theorem}
\label{thm:model_based_proj}
$\lraprojf{\lambda}$ is a Model Based Projection.
\end{theorem}

Note that $\lraprojf{\lambda}$ is linear in the size of
$\lambda$. An MBP for LIA can be defined similarly (see Appendix) based on
Cooper's method~\cite{cooper}.

\paragraph{\textbf{Bounded Safety with MBP}.} % Assume that an MBP
% $\projf{\lambda}$ is given for LRA.
Intuitively, each quantifier-free
formula in the image of $\projf{\lambda}$ under-approximates $\lambda$. As
above, we use $\lambda_m$ for the quantifier-free matrix of $\lambda$. We modify
the side-condition $\psi=\lambda$ of \textsc{Reach} and \textsc{Query} to use
quantifier-free under-approximations as follows: (i) for \textsc{Reach},
the new side-condition is $\psi = \projf{\lambda}(M)$ where $M \models
\lambda_m \land \varphi$, (ii) for \textsc{Query}, the new side-condition
is $\psi = \projf{\lambda}(M)$ where $M \models \lambda_m \land
\sem{\sum{R}(a)}^{b-1}_{\sigma}$.
Note that to avoid redundant applications of the rules, we require $M$ to satisfy a formula
stronger than $\lambda_m$.
Intuitively, (i) ensures that the newly inferred reachability fact answers the
current query and (ii) ensures that the new query cannot be immediately answered
by known facts. In both cases, the required model $M$ can be obtained as a side-effect
of discharging the premises of the rules. Soundness of \bndsafety is unaffected and
termination of $\bndsafety$ follows from the image-finiteness of $\projf{\lambda}$.
%Finally, $\bndsafety$ remains complete with the new %side-condition.
\begin{theorem}
\label{thm:termination_preserving}
Assuming an oracle and an MBP for $\thry$, \bdsafety is sound and terminating
with the modified rules.
\end{theorem}

Thus, \bndsafety with a linear-time MBP (such as $\lraprojf{\lambda}$) keeps the
size of the formulas small by efficiently inferring only the necessary
under-approximations of the quantified formulas.

%%% Local Variables:
%%% mode: latex
%%% TeX-master: "main"
%%% End:

%% file: results.tex
\section{Implementation and Experiments}
\label{sec:results}
We have implemented \textsc{RecMC} for analyzing C programs as part of
the tool \spacer. The back-end is based on Z3~\cite{z3_code} which is used
for SMT-solving and interpolation. It supports propositional logic,
linear arithmetic, and bit-vectors (via bit-blasting). The front-end
is based on \textsc{UFO}~\cite{ufo_svcomp}. It converts C programs to
the Horn-SMT format of Z3, which corresponds to the logical program
representation of Section~\ref{sec:prelims}. The
implementation and benchmarks are available
online\footnote{\url{http://www.cs.cmu.edu/~akomurav/projects/spacer/home.html}.}.

We evaluated \spacer on two sets of benchmarks. The first set contains
2,908 Boolean Programs obtained from the SLAM
toolkit\footnote{\url{https://svn.sosy-lab.org/software/sv-benchmarks/trunk/clauses/BOOL/slam.zip}}. The second
contains 799 C programs from the Software Verification
Competition (SVCOMP) 2014~\cite{svcomp14}. We call this set \textsc{Svcomp-1}. We also
evaluated on two variants of \textsc{Svcomp-1}, which we call \textsc{Svcomp-2} and
\textsc{Svcomp-3}, obtained by factoring out parts of the program into
procedures and introducing more modularity. We compared \spacer against the
implementation of GPDR in Z3. We used a time limit of 30 minutes and a memory
limit of 16GB, on an Ubuntu machine with a 2.2 GHz AMD Opteron(TM) Processor 6174
and 516GB RAM. The results are summarized in Fig.~\ref{fig:num_solved}. 
Since there are programs verified by only one of the tools,
Fig.~\ref{fig:num_solved} also reports the number of programs verified
by at least one, \ie the Virtual Best Solver (VBS).

%In general, \spacer and Z3 are complementary:
%there are programs verified by only one of the tools.

% svcomp1 - linear
% svcomp2 - outline edges
% svcomp3 - outline loops

\begin{figure}[t]
\centering
\scriptsize
\begin{tabular}{|l|c|c|c|c|c|c|c|c|}
\hline
    & \multicolumn{2}{c|}{\textsc{Slam}}
    & \multicolumn{2}{c|}{\textsc{Svcomp-1}}
    & \multicolumn{2}{c|}{\textsc{Svcomp-2}}
    & \multicolumn{2}{c|}{\textsc{Svcomp-3}} \\
\cline{2-9}
        & SAFE & UNSAFE & SAFE  & UNSAFE & SAFE & UNSAFE & SAFE & UNSAFE \\
\hline
\spacer & 1,721 & 985    & 249   & 509    & 213  & 497    & 234  & 482    \\
Z3      & 1,722 & 997    & 245   & 509    & 208  & 493    & 234  & 477    \\
VBS     & 1,727 & 998    & 252   & 509    & 225  & 500    & 240  & 482    \\
\hline
\end{tabular}
\caption{Number of programs verified by \spacer, Z3 and the Virtual Best Solver.}
\label{fig:num_solved}
\end{figure}

\subheading{Boolean Program Benchmarks} On most of the SLAM benchmarks, the runtimes of \spacer and Z3 are
similar (within 2 minutes). We then evaluated on a Boolean Program
from~\cite{bebop} in which the size of the call-tree grows exponentially in the number of
procedures. As Fig.~\ref{fig:bebop_plot} shows, \spacer handles the increasing
complexity in the example significantly better than Z3.

\begin{figure}[t]
\begin{subfigure}[b]{.5\textwidth}
\centering
\includegraphics[scale=.6]{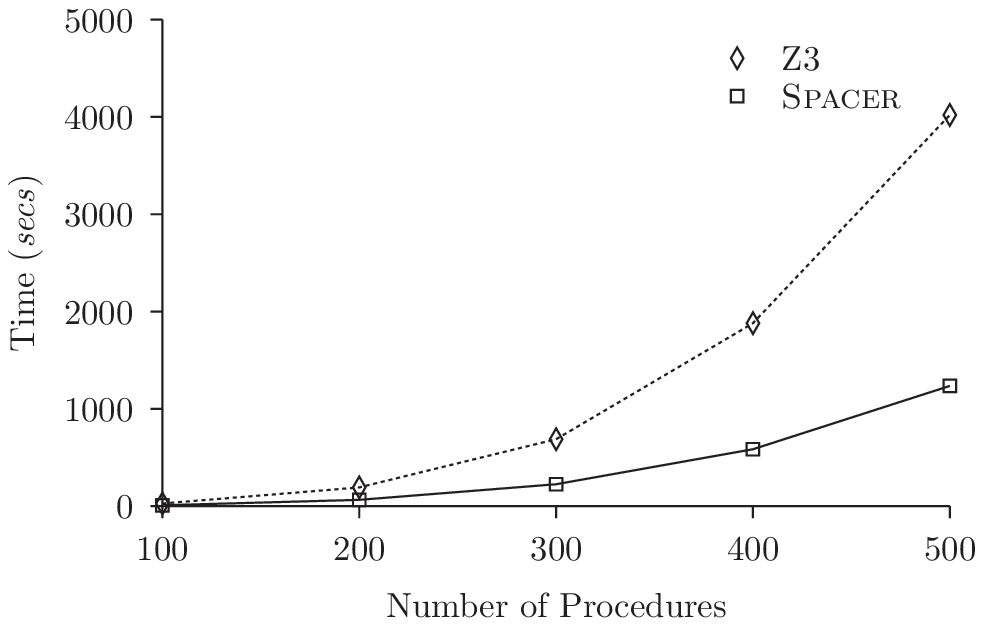}
\caption{}
\label{fig:bebop_plot}
\end{subfigure}
\begin{subfigure}[b]{.5\textwidth}
\centering
\includegraphics[scale=.6]{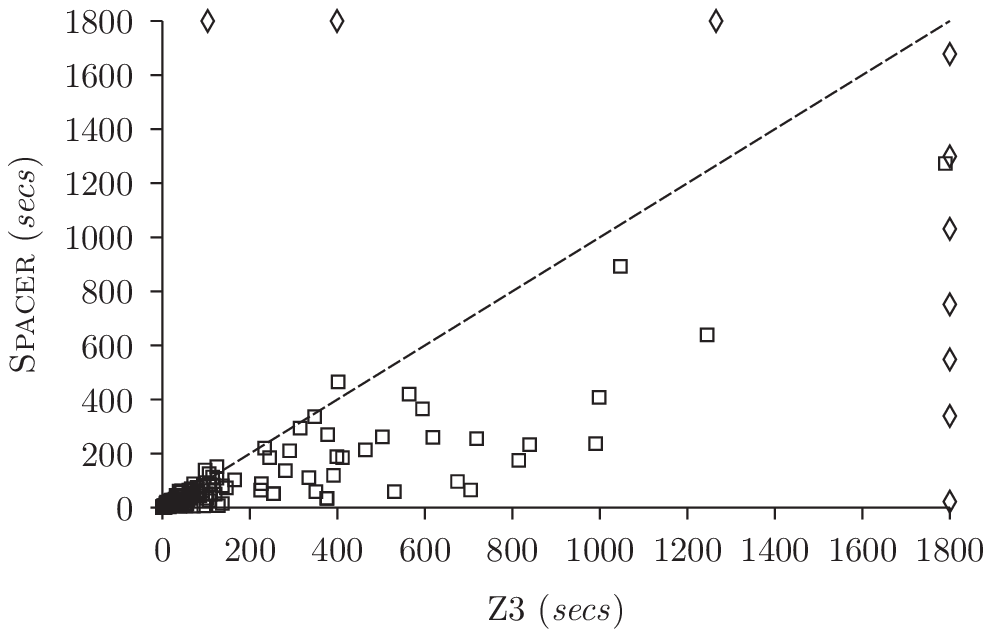}
\caption{}
\label{fig:linear_plot}
\end{subfigure}
\caption{\spacer vs. Z3 for (a) \textsc{Bebop} example and (b) \textsc{Svcomp-1} benchmarks.}
\label{}
\end{figure}

\begin{figure}[t]
\begin{subfigure}[b]{.5\textwidth}
\centering
\includegraphics[scale=.6]{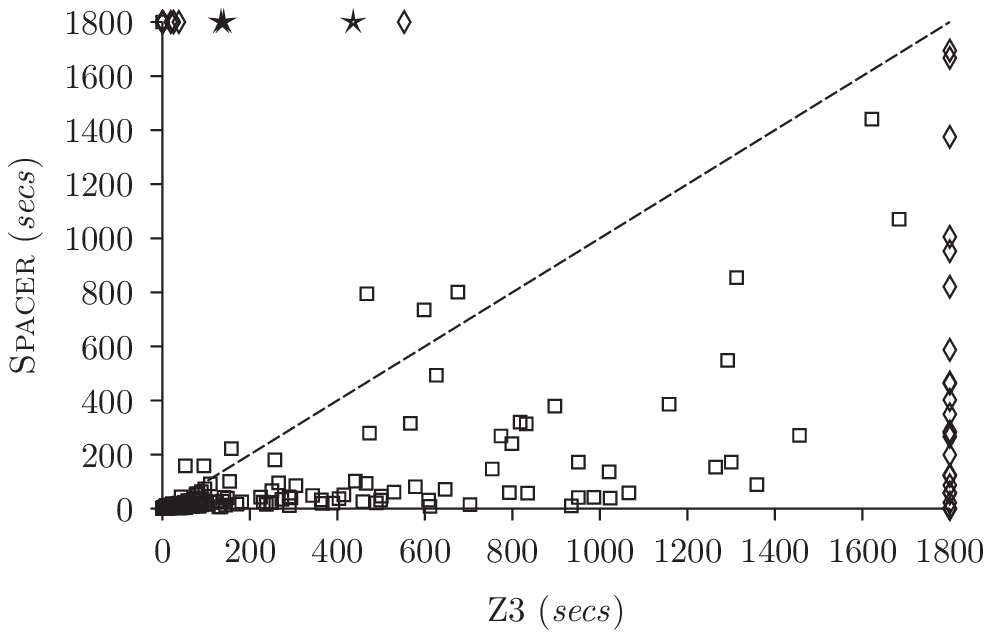}
\caption{}
\label{fig:spacer_abs_plot}
\end{subfigure}
\begin{subfigure}[b]{.5\textwidth}
\centering
\includegraphics[scale=.6]{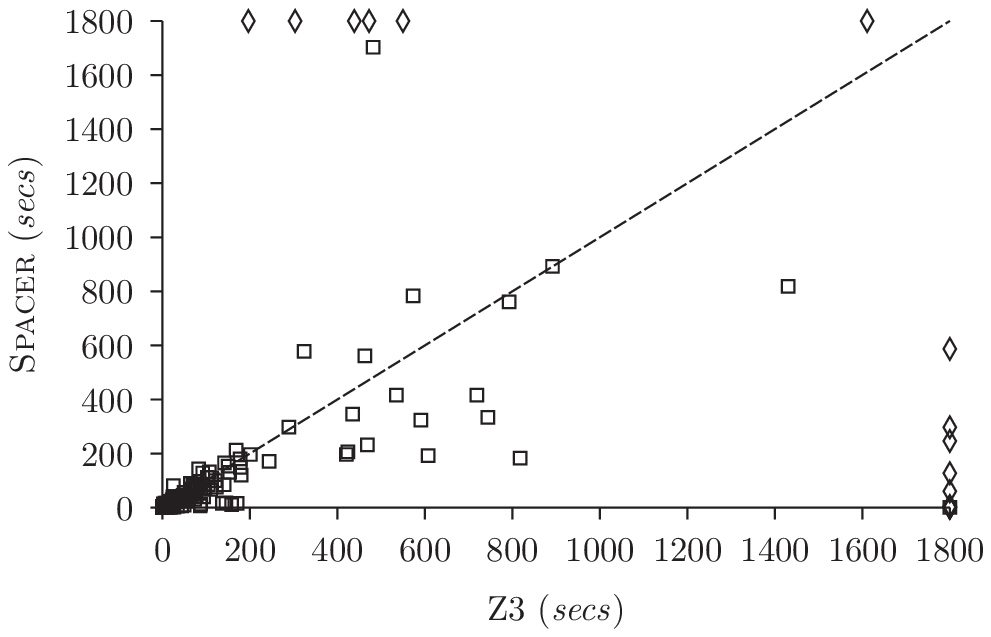}
\caption{}
\label{fig:outlined_plot}
\end{subfigure}
\caption{\spacer vs. Z3 for the benchmarks (a) \textsc{Svcomp-2} and (b)
\textsc{Svcomp-3}.}
\label{}
\end{figure}

\subheading{SVCOMP 2014 Benchmarks}
Fig.~\ref{fig:linear_plot}, \ref{fig:spacer_abs_plot}
and~\ref{fig:outlined_plot} show the scatter plots for \textsc{Svcomp-1},
\textsc{Svcomp-2} and \textsc{Svcomp-3} benchmarks. A diamond indicates a
time-out and a star indicates a mem-out. The plots show that
\spacer is significantly better on most of the programs. This shows the
practical advantage of the approximations and MBP of \recmc.
%Furthermore, as Fig.~\ref{fig:num_solved} shows, \spacer verifies strictly more programs in each
%of these benchmark sets.
%As the plot in Fig.~\ref{fig:spacer_abs_plot} and~\ref{fig:outlined_plot} show, there are
%several programs where \spacer times-out but are really easy for Z3. We traced
%down this behavior to the strategy used in \spacer for partitioning the path
%$\pi$ in the rule \textsc{Query}, when creating a new query. We observed that
%these programs become really easy for \spacer as well, when a different strategy
%is used.

%%% Local Variables:
%%% mode: latex
%%% TeX-master: "main"
%%% End:

%% file: related.tex
\section{Related Work}

There is a large body of work on interprocedural program analysis. It
was pointed out early on that verification of recursive programs is
reducible to the computation of a fixed-point over relations (called
\emph{summaries}) representing the input-output behavior of each
procedure~\cite{ed_paper,sharir81}. Such procedure summaries are called
\emph{partial correctness relations} in~\cite{ed_paper}, and are part
of the \emph{functional approach} of~\cite{sharir81}. Reps, Horwitz, and
Sagiv~\cite{rhs} showed that for a large class of finite
interprocedural dataflow problems the summaries can be computed in time
polynomial in the number of \emph{facts} and procedures. Ball and
Rajamani~\cite{bebop} adapted the RHS algorithm to the verification of
Boolean Programs.
Following the SLAM project, other software model checkers -- such as
\textsc{blast}~\cite{blast02} and \textsc{magic}~\cite{magic04} --
also implemented the CEGAR loop with predicate abstraction. None used
under-approximations of procedure semantics as we do.

Recently, several SMT-based algorithms have been proposed for safety
verification of recursive programs, including
\textsc{Whale}~\cite{whale}, HSF~\cite{hsfc},
Duality~\cite{duality}, Ultimate
Automizer~\cite{DBLP:conf/tacas/HeizmannCDEHLNSP13}, and
Corral~\cite{DBLP:conf/cav/LalQL12}. While these algorithms have been
developed independently, they share a similar structure. They use
SMT-solvers to look for  counterexamples and interpolation to
over-approximate summaries.
%
%
% First, a bounded safety
% problem is reduced to an SMT problem. Second, an SMT-solver is used to
% check for counterexamples. Third, if a counterexample is not found,
% interpolation is used to over-approximate procedure summaries. This
% process is repeated until the over-approximation of summaries becomes
% inductive, a counterexample is found, or resources are exhausted.
Corral is an exception, which relies on user input and heuristics to supply
the summaries. The algorithms differ in the SMT encoding and the
heuristics used. However, in the worst-case, they completely unroll the
call graph into a tree.
%However, for Boolean Programs they are worst-case
%exponential in the number of procedures. % In contrast, $\recmc$ uses
% rechability facts to avoid complete unrolling.

The work closest to ours is Generalized Property Driven Reachability
(GPDR) of Hoder and Bj{\o}rner~\cite{gpdr}. GPDR extends the hardware
model checking algorithm IC3 of Bradley~\cite{ic3} to SMT-supported
theories and recursive programs. Unlike $\recmc$, GPDR does not
maintain reachability facts. In the context of
Fig.~\ref{fig:basic_algo}, this means that $\rho$ is always empty and
there is no \textsc{Reach} rule. Instead, the \textsc{Query} rule is
modified to use a model $M$ that satisfies the premises (instead of our
use of the path $\pi$ when creating a query). Furthermore, the answers to the queries are
cached. In the context of Boolean Programs, this ensures that every
query is asked at most once (and either cached or blocked by a summary
fact). Since there are only finitely many
models, the algorithm always terminates. However, in the case of Linear
Arithmetic, a formula can have infinitely many models and GPDR might
end up applying the \textsc{Query} rule indefinitely. In contrast,
$\recmc$ creates only finitely many queries for a given bound on the call-stack
and is guaranteed to find a counterexample if one exists.

Combination of over- and under-approximations for analysis of
procedural programs has also been explored
in~\cite{DBLP:conf/atva/GurfinkelWC08,smash}. However, our notion of an
under-approximation is very
different. Both~\cite{DBLP:conf/atva/GurfinkelWC08,smash}
under-approximate summaries by \emph{must transitions}. A must
transition is a pair of formulas $\langle \varphi, \psi \rangle$ that
under-approximates the summary of a procedure $P$ iff for every state that
satisfies $\varphi$, $P$ has an execution that ends in a state
satisfying $\psi$. In contrast, our reachability facts are similar to
\emph{summary edges} of RHS~\cite{rhs}.  A reachability fact is a
single formula $\varphi$ such that every satisfying assignment to
$\varphi$ captures a terminating execution of $P$.

%%% Local Variables:
%%% mode: latex
%%% TeX-master: "main"
%%% End:

%% file: conclusion.tex
\section{Conclusion}
\label{sec:conclusion}

We presented \recmc, a new SMT-based algorithm for model checking
safety properties of recursive programs. For programs and properties
over decidable theories, \recmc is guaranteed to find a counterexample
if one exists. To our knowledge, this is the first SMT-based algorithm
with such a guarantee while being polynomial for Boolean Programs. The
key idea is to use a combination of under- and over-approximations of
the semantics of procedures, avoiding re-exploration of parts of the
state-space. We described an efficient instantiation of \recmc for
Linear Arithmetic (over rationals and integers) by introducing
\emph{Model-Based Projection} to under-approximate the expensive
quantifier elimination.  We have implemented it in our tool
\spacer and shown empirical evidence that it significantly improves
on the state-of-the-art.

In the future, we would like to explore extensions to other
theories. Of particular interest are the theory EUF of uninterpreted
functions with equality and the theory of arrays. The challenge is to
deal with the lack of quantifier elimination. Another direction of interest is
to combine $\recmc$ with \emph{Proof-based
Abstraction}~\cite{pba1,pba2,spacer_cav} to explore a combination of the
approximations of procedure semantics with transition-relation abstraction.

%%% Local Variables:
%%% mode: latex
%%% TeX-master: "main"
%%% End:

%% file: appendix.tex
\newpage
\appendix
\newcounter{oldtheorem}
\input{gpdr-divergence}

\section{Soundness of \recmc and \bndsafety (Proof of Theorem~\ref{thm:soundness})}
We first restate the theorem.

\setcounter{oldtheorem}{\thetheorem}
\setcounter{theorem}{0}
\begin{theorem}
\recmc and \bndsafety are sound.
\end{theorem}
\setcounter{theorem}{\theoldtheorem}

\begin{proof}
We only show the soundness of \bndsafety; the soundness of \recmc easily
follows. In particular, for
$\bndsafety(M,\varphi_\mathit{safe},n,\emptyset,\emptyset)$ we show the
following:
\begin{enumerate}
    \item if the premises of \textsc{Unsafe} hold, $M \not\models_n \varphi_\mathit{safe}$, and
    \item if the premises of \textsc{Safe} hold, $M \models_n \varphi_\mathit{safe}$.
\end{enumerate}

%As in Section~\ref{sec:prelims}, let $\{E^i\}_{i \ge -1}$ be a
%sequence of environments s.t. $E^{-1} : \sum{P} \mapsto
%\emptyset$ and $E^b : \sum{P} \mapsto \{\vec{a} \mid
%\vale{E^{b-1}}{\exists \locals{P} \st \body{P}[\formals{P} \gets
  %\vec{a}]}\}$ for $b \ge 0$. Recall that $\sem{P}^b = E^b(\sum{P})$.
%Note that, for a first-order formula $\tau$ with free variables
%$\vec{v}$ and a first-order definable environment $E$, $\sem{\tau}_E \equiv \{\vec{a}
%\mid \vale{E}{\tau[\vec{v} \gets \vec{a}]}\}$.

It suffices to show that the environments $U_\rho^b$ and $O_\sigma^b$,
respectively, under- and over-approximate the bounded semantics of the
procedures, for every $0 \le b \le n$. In particular, we show that the following is an
invariant of \bndsafety: for every model $I$ of the background theory $\thry$,
for every $Q \in \Prg$ and $b \in [0,n]$,
\begin{equation}
  \label{eq:bndsafety-invar}
  I(U_\rho^{b}(\sum{Q})) \subseteq \sem{Q}_I^b \subseteq I(O_\sigma^b(\sum{Q})).
\end{equation}

Initially, $\rho$ and $\sigma$ are empty and the invariant holds trivially.
\bndsafety updates $\sigma$ and $\rho$ in the rules \textsc{Sum} and
\textsc{Reach}, respectively. We show that these rules
preserve~\eqref{eq:bndsafety-invar}. We only show the case of
\textsc{Sum}. The case of \textsc{Reach} is similar.

Let $\langle P, \varphi, b \rangle \in \nodes$ be such that
\textsc{Sum} is applicable, \ie $\models \sem{\body{P}}_\sigma^{b-1} \implies
\neg\varphi$. Let $\psi = \textsc{Itp} (\sem{\body{P}}_\sigma^{b-1}
\comma \neg\varphi)$. Note that $\varphi$, and hence $\psi$, does not depend on
the local variables $\locals{P}$. Hence, we know that
\begin{equation}
    \label{eq:bndsafety-itp}
    \models \left(\exists \locals{P} \st \sem{\body{P}}_\sigma^{b-1}\right)
    \implies \psi.
\end{equation}
The case of $b=0$ is easy and we will skip it.
Let $I$ be an arbitrary model of $\thry$.  Assume
that~\eqref{eq:bndsafety-invar} holds at $b-1$ before applying the rule. In
particular, assume that for all $Q \in \Prg$, $\sem{Q}_I^{b-1} \subseteq
I(O_\sigma^{b-1}(\sum{Q}))$.
  
We will first show that the new summary fact $\psi$ over-approximates
$\sem{P}_I^b$. Let $J(I,\vec{X})$ be an expansion of $I$ as defined in
Section~\ref{sec:prelims}.
  
\begin{align*}
    \sem{P}_I^b &= J(I,\langle \sem{P_0}^{b-1}_I, \dots, \sem{P_n}^{b-1}_I \rangle)
                        (\exists \locals{P_i} \st \body{P_i}) &\\
                &= J(I,\langle I(O_\sigma^{b-1}(\sum{P_0})), \dots, I(O_\sigma^{b-1}(\sum{P_n})) \rangle)
                        (\exists \locals{P_i} \st \body{P_i}) &
                    (\text{hypothesis})\\
                &= I(\sem{\exists \locals{P} \st 
                    \body{P}}_{O^{b-1}_\sigma}) &
                    (\text{$O^{b-1}_\sigma$ is FO-definable})\\
                &= I(\exists \locals{P} \st \sem{\body{P}}_{O^{b-1}_\sigma})
                    & (\text{logic})\\
                &= I(\exists \locals{P} \st \sem{\body{P}}_\sigma^{b-1})
                    & (\text{notation})\\
                &\subseteq I(\psi) & (\text{from (\ref{eq:bndsafety-itp})})
\end{align*}
Next, we show that the invariant continues to hold. The map of summary facts is
updated to $\sigma' = \sigma \cup \{\langle P,b \rangle \mapsto \psi\}$. Now,
$\sigma'$ differs from $\sigma$ only for the procedure $P$ and every bound in
$[0,b]$. Let $b' \in [0,b]$ be arbitrary. Since~\eqref{eq:bndsafety-invar} was
true before applying $\textsc{Sum}$, we know that $\sem{P}_I^{b'} \subseteq
I(O_\sigma^{b'}(\sum{P}))$. As $\sem{P}_I^{b'} \subseteq \sem{P}_I^b
\subseteq I(\psi)$, it follows that $\sem{P}_I^{b'} \subseteq
I(O_{\sigma}^{b'}(\sum{P})) \cap I(\psi) \subseteq
I(O_{\sigma}^{b'}(\sum{P}) \land \psi) = I(O_{\sigma'}^{b'}(\sum{P}))$.
\qed
\end{proof}

\section{Termination of \bndsafety (Proof of Theorem~\ref{thm:termination})}
We first restate the theorem:
\setcounter{oldtheorem}{\thetheorem}
\setcounter{theorem}{1}
\begin{theorem}
Given an oracle for $\thry$, $\bndsafety(\Prg, \varphi, n,
\emptyset, \emptyset)$ terminates.
%Let $\Prg$ be a program, $\varphi$ be a safety property, $p$ be the maximum
%number of paths in $P \in \Prg$, $c$ be the maximum number of procedure calls
%along any path in $\Prg$, and $N$ be the number of procedures in $\Prg$.  Assume
%an oracle for satisfiability modulo $\thry$. Then, $\bndsafety(\Prg, \varphi, n,
%\emptyset, \emptyset)$ terminates in $O((N \cdot p)^{O(n^2)} \cdot n \cdot
%c^n)$-many applications of the rules in Fig.~\ref{fig:basic_algo}.
\end{theorem}
\setcounter{theorem}{\theoldtheorem}

We begin with showing some useful lemmas. Let $p$ be the maximum number of paths
in $P \in \Prg$, $c$ be the maximum number of procedure calls along any path in
$\Prg$, $N$ be the number of procedures in $\Prg$ and assume an oracle for SAT
modulo $\thry$.

The following lemma shows that when a query is removed from $\nodes$, it is
actually answered. The proof is immediate from the definitions of $O_\sigma^b$
and $U_\rho^b$ given in Section~\ref{sec:prelims}.

\begin{lemma}[Answered Queries]
\label{lem:relevant_facts}
Whenever \bndsafety removes a query from $\nodes$, it is answered using the
known summary and reachability facts. In particular, let
\textsc{Sum} or \textsc{Reach} be applied to $\langle P, \varphi, b \rangle$.
Then, for every $\langle P, \eta, b \rangle \in \nodes$ removed from $\nodes$ by
the rule,
\begin{enumerate}
    \item after \textsc{Sum} is applied, $\models \sem{\sum{P}}_\sigma^b
        \implies \neg\eta$,
        and
    \item after \textsc{Reach} is applied, $\not\models \sem{\sum{P}}_\rho^b
        \implies \neg\eta$.
\end{enumerate}
\end{lemma}

Next, we show that new facts need to be inferred to answer queries remaining in
$\nodes$.

\begin{lemma}[Pending Queries]
\label{lem:reuse}
$\nodes$ only has the queries which
cannot be immediately answered by $\rho$ or $\sigma$, \ie
as long as $\langle P, \eta, \ell \rangle$ is in $\nodes$, the following
are invariant across iterations of \bndsafety.
    \begin{enumerate}
        \item $\not\models \sem{\sum{P}}_\sigma^{\ell} \implies \neg\eta$, and
        \item $\models \sem{\sum{P}}_\rho^{\ell} \implies \neg\eta$.
    \end{enumerate}
\end{lemma}

\begin{proof}
We first show that the invariants hold when a query is newly created by {\sc
Query}. Let $P$, $\eta$ and $\ell$ be, respectively, $R$, $\psi[\vec{a} \leftarrow
\formals{R}]$ and $b-1$, as in the conclusion of the rule.
The last-but-one premise of \textsc{Query} is
\[
    \models \sem{\pi_u}_\sigma^{b-1} \land \sem{\sum{R}(\vec{a})}_\rho^{b-1} \land \sem{\pi_v}_\rho^{b-1}
    \implies \neg\varphi
\]
which implies that
\[
    \models \sem{\sum{R}(\vec{a})}_\rho^{b-1} \implies \neg \left( \sem{\pi_u}_\sigma^{b-1}
    \land \sem{\pi_v}_\rho^{b-1} \land \varphi \right).
\]
The variables not in common, \viz $(\formals{P} \cup \locals{P}) \setminus
\vec{a}$, can be universally quantified from the right hand side resulting in
$\models \sem{\sum{R}}_\rho^{b-1} \implies \neg\eta$.  Similarly, $\not\models
\sem{\sum{R}}_\sigma^{b-1} \implies \neg\eta$ follows from the last premise of
the rule. Next, we show that \textsc{Sum} and \textsc{Reach} preserve the
invariants.

Let \textsc{Sum} answer a query $\langle P, \varphi, \ell \rangle$ with a new
summary fact $\psi$ and let the updated map of summary facts be $\sigma' =
\sigma \cup \{\langle P,\ell \rangle \mapsto \psi\}$. Now, consider $\langle P,
\eta, \ell' \rangle \in \nodes$ after the application of the rule. If $\ell' >
\ell$, $O^{\ell'}_{\sigma'} = O^{\ell'}_\sigma$ and the invariant continues to
hold.  So, assume $\ell' \le \ell$. From the conclusion of \textsc{Sum}, we have
$\not\models \sem{\sum{P}}^{\ell'}_\sigma \land \psi \implies \neg\eta$.  Now,
$\sem{\sum{P}}^{\ell'}_{\sigma'} = \sem{\sum{P}}^{\ell'}_\sigma \land \psi$. So,
the invariant continues to hold.

Similarly, let \textsc{Reach} answer a query $\langle P, \varphi, \ell \rangle$
with a new reachability fact $\psi$ and let the updated map of reachability
facts be $\rho' = \rho \cup \{\psi \mapsto \langle P,\ell \rangle\}$. Now,
consider $\langle P, \eta, \ell' \rangle \in \nodes$ after the application of
the rule. If $\ell' < \ell$, $U^{\ell'}_{\rho'} = U^{\ell'}_\rho$ and the
invariant continues to hold. So, assume $\ell' \ge \ell$.  From the conclusion
of \textsc{Reach}, we have $\models \psi \implies \neg\eta$. Assuming the
invariant holds before the rule application, we also have $\models
\sem{\sum{P}}_\rho^{\ell'} \implies \neg\eta$. Therefore, we have $\models
\sem{\sum{P}}_\rho^{\ell'} \lor \psi \implies \neg\eta$. Now,
$\sem{\sum{P}}^{\ell'}_{\rho'} = \sem{\sum{P}}^{\ell'}_\rho \lor \psi$. So, the
invariant continues to hold.
\qed
\end{proof}

The next few lemmas show that the rules of the algorithm cannot be applied
indefinitely, leading to a termination argument. Let $N$ be the number of
procedures in the program, $p$ be the maximum number of paths in a procedure,
$c$ be the maximum number of procedure calls along any path in $\Prg$ and $n$ be
the current bound on the call-stack.

\begin{lemma}[Finite Reach Facts]
\label{lem:finite_reach}
The environment $U^b_\rho$ is updated for a given predicate symbol $\sum{P}$ and
a bound $b$ on the call-stack only $O(N^b \cdot p^{b+1})$-many times.
\end{lemma}

\begin{proof}
The environment $U^b_\rho$ can be updated for $\sum{P}$ and $b$ whenever a
reachability fact is inferred for $P$ at a bound $b' \le b$. Now, a reachability
fact is obtained per path (after eliminating the local variables) of a
procedure, using the currently known reachability facts about the callees.
Moreover, Lemmas~\ref{lem:relevant_facts} and~\ref{lem:reuse} imply that no
reachability fact is inferred twice. This is because whenever a query is
answered using \textsc{Reach}, the query could not have been answered using
already existing reachability facts and a new reachability fact is inferred.

This gives the following recurrence $\mathit{Reach}(b)$ for the number of
updates to $U^b_\rho$ for a given $\sum{P}$:
\[
    \mathit{Reach}(b) =
        \begin{cases}
                                 p, & b = 0\\
            (p \cdot N + 1) \cdot \mathit{Reach}(b-1), & b > 0.
        \end{cases}
\]
In words, for $b=0$, the number of updates is given by the number of
reachability facts that can be inferred, which is bounded by the number of paths
$p$ in the procedure $P$. For $b>0$, the environment $U^{b-1}_\rho$ is updated
when a reachability fact is learnt at $b$ or at a bound smaller than $b$. The
latter is simply $\mathit{Reach}(b-1)$. For the former, a new reachability fact
is inferred at $b$ along a path whenever $U^{b-1}_\rho$ changes for a callee.
For $N$ procedures and $p$ paths, this is given by $(p \cdot N \cdot
\mathit{Reach}(b-1))$.

This gives us $\mathit{Reach}(b) = O(N^b \cdot p^{b+1})$.
\qed
\end{proof}

\begin{lemma}[Finite Queries]
\label{lem:finite_query}
For $\node{P}{\varphi}{b} \in \nodes$, \textsc{Query} is applicable only
$O(c \cdot N^b \cdot p^{b+1})$-many times.
\end{lemma}

\begin{proof}
First, assume that the environments $U^{b-1}_\rho$ and $O^{b-1}_\sigma$ are
fixed.  The number of possible queries that can be created for a given path of
$P$ is bounded by the number of ways the path can be divided into a prefix, a
procedure call, and a suffix. This is bounded by $c$, the maximum number of
calls along the path. For $p$ paths, this is bounded by $c \cdot p$.

Consider a path $\pi$ and its division, and let a query be created for a callee
$R$ along $\pi$.  Now, while the query is still in $\nodes$, updates to the
environments $O^{b-1}_\sigma$ and $U^{b-1}_\rho$ do not result in a new query
for $R$ for the same division along $\pi$. This is because, the new query would
overlap with the existing one and this is disallowed by the second
side-condition of \textsc{Query}.

Suppose that the new query is answered by \textsc{Sum}. With the updated map of
summary facts, the last premise of \textsc{Query} can be shown to fail for the
current division of $\pi$. If $O^{b-1}_\sigma$ is updated, the last premise
continues to fail. So, a new query can be created for the same prefix and suffix
along $\pi$ only if $U^{b-1}_\rho$ is updated for some callee along $\pi$. The
other possibility is that the query is answered by \textsc{Reach} which updates
$U^{b-1}_\rho$ as well.

Thus, for a given path, and a given division of it into prefix and suffix, the
number of queries that can be created is bounded by the number of updates to
$U^{b-1}_\rho$ which is $(N \cdot \mathit{Reach}(b-1))$. Here, $\mathit{Reach}$
is as in Lemma~\ref{lem:finite_reach}. So, the number of
times \textsc{Query} is applicable for a given query $\node{P}{\varphi}{b}$ is
$O(p \cdot c \cdot N \cdot \mathit{Reach}(b-1))$. As $\mathit{Reach}(b) = N^b
\cdot p^{b+1}$, we obtain the bound $O(c \cdot N^b \cdot p^{b+1})$.
%
%Assume that $\sigma$ is fixed for the callees at $b-1$. The number
%of queries that can be created for a callee $Q$ is bounded by the number of
%reachability facts that can be inferred for $\langle Q,b-1 \rangle$. This
%results from the second side-condition of \textsc{Query} that only
%non-overlapping queries can be in $\nodes$ at any given time (and from
%Lemma~\ref{lem:relevant_facts} that a newly inferred reachability fact witnesses
%the query just answered). The number of reachability facts for $\langle Q,b-1
%\rangle$ is $O(\mathit{Reach}(b-1))$, due to Lemma~\ref{lem:finite_reach}.
%For up to $N$ callees, the bound is $O(N \cdot \mathit{Reach}(b-1))$.
%
%%Now, $\rho$ of callees at $b-1$ can only change $O(N \cdot \mathit{Reach}(b-1))$ times, again by
%%Lemma~\ref{lem:finite_reach}.
%Changing $\sigma$ only strengthens $O_\sigma^{b-1}$. The only way this
%can affect the number of queries is by creating a stronger query than one
%already in $\nodes$. But this is disallowed by the second side-condition of
%\textsc{Query}.
%
%So, the number of times \textsc{Query} is applicable is $O(N \cdot
%\mathit{Reach}(b-1)) = O(N^{b}p^{b})$.
\qed
\end{proof}

\begin{lemma}[Progress]
\label{lem:progress}
As long as $\nodes$ is non-empty in \bndsafety, either \textsc{Sum},
\textsc{Reach} or \textsc{Query} is always applicable.
\end{lemma}

\begin{proof}
First, we show that for every query in $\nodes$, either of the three rules is
applicable, without the second side-condition in \textsc{Query}.
Let $\langle P,\varphi,b \rangle \in \nodes$.
If $\models \sem{\body{P}}_\sigma^{b-1} \implies \neg\varphi$, then \textsc{Sum} is applicable. Otherwise, there exists
a path $\pi \in \paths{P}$ such that $\sem{\pi}_\sigma^{b-1}$ is satisfiable with
$\varphi$, \ie $\not\models \sem{\pi}_\sigma^{b-1} \implies \neg\varphi$. Now, if
$\sem{\pi}_\rho^{b-1}$ is also satisfiable with $\varphi$, \ie $\not\models
\sem{\pi}_\rho^{b-1} \implies \neg\varphi$, \textsc{Reach} is applicable. Otherwise,
$\models \sem{\pi}_\rho^{b-1} \implies \neg\varphi$. Note that this can only happen if
$b > 0$, as otherwise, there will not be any procedure calls along $\pi$ and
$\sem{\pi}_\sigma^{b-1}$ and $\sem{\pi}_\rho^{b-1}$ would be equivalent.

Let $\pi = \pi_0 \land \pi_1 \land \dots \pi_l$ for some finite $l$. Then, $\sem{\pi}_\sigma^{b-1}$ is
obtained by taking the conjunction of the formulas
\[
    \langle \sem{\pi_0}_\sigma^{b-1}, \sem{\pi_1}_\sigma^{b-1}, \dots \rangle.
\]
Similarly, $\sem{\pi}_\rho^{b-1}$ is obtained by taking the conjunction of the
formulas
\[
    \langle \sem{\pi_0}_\rho^{b-1}, \sem{\pi_1}_\rho^{b-1}, \dots \rangle.
\]
From Theorem~\ref{thm:soundness}, we can think of obtaining the latter sequence of
formulas by conjoining $\sem{\pi_i}_\rho^{b-1}$ to $\sem{\pi_i}_\sigma^{b-1}$
for every $i$. When this is done backwards for decreasing values of $i$, an intermediate sequence looks like
\[
    \langle \sem{\pi_0}_\sigma^{b-1}, \dots, \sem{\pi_{j-1}}_\sigma^{b-1},
    \sem{\pi_j}_\rho^{b-1} \dots \rangle.
\]
As $\sem{\pi}_\rho^{b-1}$ is unsatisfiable with $\varphi$, there
exists a maximal $j$ such that the conjunction of constraints in such an
intermediate sequence are unsatisfiable with $\varphi$. Moreover, $\pi_j$ must be
a literal of the form $\sum{R}(\vec{a})$ as otherwise, $\sem{\pi_j}_\sigma^{b-1} =
\sem{\pi_j}_\rho^{b-1}$ violating the maximality condition on $j$.
Thus, all premises of \textsc{Query} hold and the rule is applicable.

Now, the second side-condition in \textsc{Query} can be trivially satisfied by
always choosing a query in $\nodes$ with the smallest bound for the next rule to
apply. This is because, if $\node{R}{\eta}{b-1}$ is the newly created query, there
is no other query in $\nodes$ for $R$ and $b-1$.
\qed
\end{proof}

Lemmas~\ref{lem:finite_query} and~\ref{lem:progress} imply that every query in
$\nodes$ is eventually answered by \textsc{Sum} or \textsc{Reach}, as shown
below.

\begin{lemma}[Eventual Answer]
\label{lem:eventual_answer}
Every $\node{P}{\varphi}{b} \in \nodes$ is eventually answered by \textsc{Sum} or
\textsc{Reach}, in $O(b \cdot c^b \cdot (Np)^{O(b^2)})$ applications
of the rules.
\end{lemma}

\begin{proof}
Firstly, to answer any given query in $\nodes$, Lemma~\ref{lem:finite_query} guarantees that the algorithm can only
create finitely many queries. Lemma~\ref{lem:progress} guarantees that some rule
is always applicable, as long as $\nodes$ is non-empty. Thus, when \textsc{Query} cannot be applied for
any query in $\nodes$, either \textsc{Sum} or \textsc{Reach} must be
applicable for some query. Thus, eventually, all queries are answered.

The total number of rule applications to answer $\node{P}{\varphi}{b}$ is then linear in
the cumulative number of applications of \textsc{Query}, which has the following recurrence:
\[
    T(b) =
        \begin{cases}
            Q(0), & b = 0\\
            Q(b) (1 + T(b-1)), & b > 0.
        \end{cases}
\]
where $Q(b)$ denotes the number of applications of \textsc{Query} for a
fixed query in $\nodes$ at bound $b$.
From Lemma~\ref{lem:finite_query}, $Q(b) = O(c \cdot N^{b} \cdot p^{b+1})$. This gives us
$T(b) = O(b \cdot c^b \cdot (Np)^{O(b^2)})$.
\qed
\end{proof}

\subheading{Main Proof} Follows immediately from
Lemma~\ref{lem:eventual_answer}.
\qed

\section{Complexity of \recmc for Boolean Programs (Proof of Theorem~\ref{thm:bool_prog})}
We first restate the theorem:
\setcounter{oldtheorem}{\thetheorem}
\setcounter{theorem}{2}
\begin{theorem}
Let $\Prg$ be a Boolean Program. Then $\recmc(\Prg, \varphi)$ terminates in
$O(N^2 \cdot 2^{2k})$-many applications of the rules in
Fig.~\ref{fig:basic_algo}.
%Let $\Prg$ be a Boolean Program, $\varphi$ a safety property, $N$ the number of
%procedures in $\Prg$ and $a$ the maximum arity of $P \in \Prg$. Then,
%$\recmc(\Prg, \varphi)$ terminates in $O(N^2 \cdot 2^{2a})$-many applications
%of the rules in Fig.~\ref{fig:basic_algo}.
\end{theorem}
\setcounter{theorem}{\theoldtheorem}

\begin{proof}
First, assume a bound $n$ on the call-stack. The number of queries that can be
created for a procedure at any given bound is
$O(2^k)$, the number of possible valuations of the parameters (note that
\textsc{Query} disallows overlapping queries to be present simultaneously in
$\nodes$). For $N$ procedures and $n$ possible values of the bound, the
complexity of $\bndsafety(\Prg, \varphi, n, \emptyset, \emptyset)$, for a Boolean
Program, is $O(N \cdot 2^k \cdot n)$.

Now, the total number of summary facts that can be inferred for a procedure is
also bounded by $O(2^k)$. As $O_\sigma^b$ is monotonic in $b$, the number
of iterations of \recmc is bounded by $O(N \cdot 2^k)$, the cumulative number of states
of all procedures. Thus, we obtain the complexity of \recmc as $O(N^2 \cdot 2^{2k})$.
%Now, the total number of reachability facts and summary facts that can be inferred
%for a procedure is also bounded by $O(2^a)$. As $\{U_\rho^b\}_b$ and
%$\{O_\sigma^b\}_b$ are monotonic (increasing and decreasing, respectively), the
%number of iterations of \recmc is bounded by $O(N2^a)$. As $\sigma$ and $\rho$
%are reused across iterations of \recmc, we obtain the complexity of $O(N^22^{2a})$.
\qed
\end{proof}

%\section{Complexity Comparison with Na\"{i}ve Unrolling}
%In the best-case, the size of $\body{P}$ for a procedure $P$ is $O(c \log p)$ where $c$ and $p$
%are as in Section~\ref{sec:algo}. This is because, $p$ can be exponential in the
%DAG representation of $G_P$ and $c$ is the maximum number of procedure calls
%along a path. If the program were na\"{i}vely unrolled up
%to the bound $n$, the best-case size of the unrolling is given by the following
%recurrence.
%\[
    %S(b) = \begin{cases}
                %c \log p, & b = 0 \\
                %(c \log p) (1 + c \cdot S(b-1)), & b > 0.
           %\end{cases}
%\]
%This gives,
%\begin{align*}
    %S(b) &= c \log p + c^3\log^2 p + \dots + c^{2(b+1)}\log^{b+1} p \\
         %&= O((b+1) \cdot c^{2(b+1)} \cdot \log^{b+1}p).
%\end{align*}
%Assuming an oracle for SMT is NP-complete, and assuming P $\neq$ NP, the
%time-complexity would be exponential in the above size, in the worst-case. This
%results in the complexity of $O(2^{(b+1) \cdot c^{2(b+1)} \cdot \log^{b+1} p})$,
%which is $O(p^{(b+1) \cdot c^{2(b+1)} \cdot \log^b p})$.
%% TODO: our complexity doesn't include the cost of QE!
%Our complexity shown in Section~\ref{sec:algo} is exponentially better.

\section{\bndsafety with MBP (Proof of Theorem~\ref{thm:termination_preserving})}
Here, we show that \bndsafety with MBP is sound and terminating.

First of all, in presence of MBP, \textsc{Sum} is unaffected and a reachability fact
inferred by \textsc{Reach} is only strengthened. Thus, soundness of \bndsafety
(Theorem~\ref{thm:soundness}) is preserved.

Then, it is easy to show that the modified side-conditions to \textsc{Reach} and
\textsc{Query} preserve Lemmas~\ref{lem:relevant_facts} and~\ref{lem:reuse} and
we skip the proof.

Then, we will show that the \emph{finite-image} property of an MBP preserves
the finiteness of the number of reachability facts inferred and the number of
queries generated by the algorithm.
Let $d$ be the size of the image of an MBP.
In the proof of Lemma~\ref{lem:finite_reach}, the recurrence relation has an
extra factor of $d$. The rest of the proof of finiteness of the number of reachability
facts remains the same. Similarly, in the proof of Lemma~\ref{lem:finite_query}, the number of
times \textsc{Query} can be applied along a path for a fixed division and fixed
environments $O^{b-1}_\sigma$ and $U^{b-1}_\rho$ increases by
a factor of $d$. Again, the rest of the proof of finiteness of the number of
queries generated remains the same. That is, Lemmas~\ref{lem:finite_reach}
and~\ref{lem:finite_query}, and hence, Lemma~\ref{lem:eventual_answer}, are
preserved with scaled up complexity bounds.

Note that Theorem~\ref{lem:progress} is unaffected by under-approximations.

Together, we have that Theorem~\ref{thm:termination} is preserved, with a scaled up
complexity bound.
\qed

\section{$\lraprojf{\lambda}$ is an MBP (Proof of Lemma~\ref{thm:model_based_proj})}
\newcommand{\sub}[1]{\text{\emph{Sub}}_t({#1})}
First, we restate the theorem:
\setcounter{oldtheorem}{\thetheorem}
\setcounter{theorem}{3}
\begin{theorem}
$\lraprojf{\lambda}$ is a Model Based Projection.
\end{theorem}
\setcounter{theorem}{\theoldtheorem}

\begin{proof}
By definition, $\lraprojf{\lambda}$ has a finite image, as there are only
finitely many disjuncts in (\ref{eq:lw}). Thus, it suffices to
show that for every $M \models \lambda_m$, $M \models \lraprojf{\lambda}(M)$.

Each disjunct in the LW decomposition (\ref{eq:lw}) is obtained by a \emph{virtual
substitution} of
the literals in $\lambda_m$ containing $x$. As in Section~\ref{sec:lazy_qe}, we
assume that $\lambda_m$ is in NNF with the only literals containing $x$ of the
form $(x=e)$, $(\ell<x)$ or $(x<u)$ for $x$-free terms $e$, $\ell$ and $u$.
Let $\mathit{Sub}_t$ denote the virtual substitution map of literals when $t$ is either $e$,
$\ell+\epsilon$ or $-\infty$. The LW method~\cite{lw} defines:
\begin{align}
    &    \mathit{Sub}_e(x=e) = \top, \mathit{Sub}_e(\ell<x) = (\ell<e), \mathit{Sub}_e(x<u) = (e<u) \\
    & \mathit{Sub}_{\ell+\epsilon}(x=e) = \bot, \mathit{Sub}_{\ell+\epsilon}(\ell'<x) = (\ell'
                            \le \ell), \mathit{Sub}_{\ell+\epsilon}(x<u) = (\ell<u)\\
    & \mathit{Sub}_{-\infty}(x=e) = \bot, \mathit{Sub}_{-\infty}(\ell<x) = \bot,
                                       \mathit{Sub}_{-\infty}(x<u) = \top
\end{align}

Let $M \models \lambda_m$ and $\lraprojf{\lambda}(M) = \lambda_m[t]$ where $t$ is
either $e$ or $\ell+\epsilon$ or $-\infty$.
As $\lambda_m$ is in NNF, it suffices to show that
for every literal $\mu$ of $\lambda_m$ containing $x$, the following holds:
\begin{equation}
    \label{eqn:mval_preserved}
    M \models (\mu \implies \mathit{Sub}_t(\mu))
\end{equation}

We consider the different possibilities of $t$ below. For a term $\eta$, let
$M[\eta]$ denote the value of $\eta$ in $M$.

\begin{description}
    \item {Case $t = e$.} In this case, we know that $M\models x=e$.
        Now, for a literal $\ell < x$,
        \begin{align*}
            M[\ell < x] &\implies M[\ell] < M[x] & \\
                        &= M[\ell] < M[e] & \\
                        &= M[\ell<e] &\\
                        &= M[\sub{\ell<x}] & \{\sub{\ell<x} = (\ell < e)\}.
        \end{align*}
        Similarly, literals of the form $x<u$ and $x=e'$ can be considered.

    \item {Case $t =\ell+\epsilon$.} In this case, we know that $M[\ell<x]$ is
        true, \ie $M[\ell] < M[x]$ and whenever $M[\ell' < x]$ is true, $M[\ell'
        \le \ell]$ is also true.
        Now, for a literal $\ell' < x$,
        \begin{align*}
            M[\ell' < x] &\implies M[\ell' \le \ell] & \\
                        &= M[\sub{\ell' < x}] & \{\sub{\ell'<x} = (\ell' \le
        \ell)\}.
        \end{align*}

        For a literal $x < u$,
        \begin{align*}
            M[x < u] &\implies M[x] < M[u] & \\
                     &\implies M[\ell] < M[u] & \{M[\ell] < M[x]\}\\
                     &\implies M[\ell<u] & \\
                     &= M[\sub{x<u}] & \{\sub{x<u} = (\ell<u)\}
        \end{align*}

        For a literal $x=e$, (\ref{eqn:mval_preserved}) vacuously holds as
        $M[x=e]$ is false.

    \item {Case $t = -\infty$.} In this case, we know that $M[x=e]$ and $M[\ell < x]$
        are false for every literal of the form $x=e$ and $\ell<x$. So, for such literals
        (\ref{eqn:mval_preserved}) vacuously holds. For
        a literal $x<u$, $\sub{x<u} = \top$ and hence, (\ref{eqn:mval_preserved})
        holds again.
\end{description}
%Let $t = T(\matrix,M)$ and $\mu = \text{\emph{LW}}(\matrix,T(\matrix,M))$. Recall that $\mu$ is obtained
%from $\matrix$ as the result of a substitution of the literals of $\matrix$
%containing $x$, as shown in Section~\ref{sec:lazy_qe}. Let the substitution map
%be denoted by $\text{\emph{Sub}}_t$. As $\matrix$ is assumed
%to be in NNF, it suffices to show that for every literal $\lambda \in
%\lits{\matrix}$ containing $x$, the following holds
%\begin{equation}
    %\label{eqn:mval_preserved}
    %M[\lambda] \implies \proj{M}{x}[\sub{\lambda}].
%\end{equation}
%In other words, if a literal is true under $M$ before the substitution, it continues to be
%true after the substitution.
%
%We consider three cases based on the definition of $T$.
\qed
\end{proof}

\input{mbplia}

%%% Local Variables:
%%% mode: latex
%%% TeX-master: "main"
%%% End:

%% file: gpdr-divergence.tex
\section{Divergence of GPDR for Bounded Call-Stack}
Consider the program $\langle \langle M,L,G \rangle, M \rangle$ with
$M = \langle y_0, y, \sum{M}, \langle x, n \rangle, \body{M}
\rangle$,
$L = \langle n, \langle x,y,i \rangle, \sum{L}, \langle x_0,y_0,i_0 \rangle,
\body{L} \rangle$,
$G = \langle x_0, x_1, \sum{G}, \emptyset, \body{G} \rangle$, where

\begin{align*}
    \body{M} &= \sum{L}(x,y_0,n,n) \land \sum{G}(x,y) \land n>0 \\
    \body{L} &= \left( i=0 \land x=0 \land y=0 \right) \lor\\
             &  \quad\left( \sum{L}(x_0,y_0,i_0,n) \land x=x_0+1 \land y=y_0+1 \land
               i=i_0+1 \land i>0 \right) \\
    \body{G} &= (x=x_0+1)
\end{align*}

The GPDR~\cite{gpdr} algorithm can be shown to diverge when checking $M
\models_2 y_0 \le y$, for \eg by inferring the diverging sequence of over-approximations of $\sem{L}^1$:
\[
(x<2 \implies y\le 1), (x<3 \implies y\le2), \dots
\]
We also observed this behavior experimentally (Z3 revision
\texttt{d548c51} at \url{http://z3.codeplex.com}). The Horn-SMT file
for the example is available at
\[
 \text{\url{http://www.cs.cmu.edu/~akomurav/projects/spacer/gpdr_diverging.smt2}}.
\]

%%% Local Variables:
%%% mode: latex
%%% TeX-master: "main"
%%% End:

%% file: mbplia.tex
\newcommand{\liaprojf}[1]{\mathit{LIAProj}_{#1}}
\section{Model Based Projection for  Linear Integer Arithmetic}
In this section, we present our MBP $\liaprojf{\lambda}$ for LIA.  It
is based on Cooper's method for Quantifier Elimination
procedure~\cite{cooper}. Let $\lambda(\vec{y}) = \exists x \st
\lambda_m (x,\vec{y})$, where $\lambda_m$ is quantifier free and in
negation normal form. Without loss of generality, let the only
literals containing $x$ be the form $\ell < x$, $x < u$, $x = e$ or
$(d \mid \pm x + w)$, where $a \mid b$ denotes that $a$ divides
$b$, the terms $\ell$, $u$, $e$ and $w$ are $x$-free, and $d \in
\mathbb{Z} \setminus \{0\}$. Let $E = \{e \mid (x = e) \in
\lits{\lambda_m}\}$ be the set of equality terms of $x$ and $L =
\{\ell \mid (\ell < x) \in \lits{\lambda_m}\}$ be the set of
lower-bounds of $x$.  Then, by Cooper's method,
\begin{equation}
  \label{eq:cooper}
    \exists x \st
\lambda_m (x,\vec{y}) \equiv \bigor_{(x=e) \in \lits{\lambda}}  \lambda_m[e] \lor                  \bigor_{(\ell < x) \in \lits{\lambda}} \left( \bigor_{i=0}^{D-1} \lambda_m[\ell + 1 + i] \right) \lor
                  \bigor_{i=0}^{D-1} \lambda_m^{-\infty}[i].
\end{equation}
where $D$ is the least common multiple of all the divisors in the
divisibility literals of $\lambda_m$, $[\cdot]$ denotes a substitution for $x$
and $\lambda_m^{-\infty}$ is
obtained from $\lambda_m$ by substituting all non-divisibility literals
as follows:
\begin{align}
(\ell < x)  &\mapsto  \bot &
(x < u)     &\mapsto  \top &
(x=e)       &\mapsto  \bot
\end{align}

Intuitively, the disjunction partitions the space of the possible values of $x$.
A disjunct for $(x=e)$ covers the case when $x$ is equal to an equality term.
Otherwise, the lower-bounds identify various intervals in which $x$ can be
present. The disjuncts for $(\ell < x)$ cover the case when $x$ satisfies a
lower-bound, and the last disjunct is for the case when $x$ is smaller than all
lower-bounds. The disjunction over the possible values of $i$ covers the
different ways in which the divisibility literals can be satisfied.

Model-based projection $\liaprojf{\lambda}$ is defined as follows,
conflicts are resolved by some arbitrary, but fixed, syntactic
ordering on terms:
\begin{equation}
    \liaprojf{\lambda}(M) = \begin{cases}
                    \lambda_m[e], & \text{if } x=e \in \lits{\lambda} \land  M\models (x=e) \\
                    \lambda_m[\ell + 1 + i_\ell], & \text{else if }
                          (\ell <x) \in \lits{\lambda} \land M \models (\ell<x)
                          \land{}\\ & \forall
                                (\ell'<x) \in \lits{\lambda} \st \left( M \models
                                  ((\ell' < x) \implies (\ell' \le \ell)) \right)\\
                    \lambda_m^{-\infty}[i_{-\infty}], & \text{otherwise}
                    \end{cases}
\end{equation}
where $i_\ell = M[x-(\ell+1)] \bmod D$,  $i_{-\infty} =
M[x] \bmod D$, and $M[x]$ is the value of $x$ in $M$.
The following lemma shows that $\liaprojf{\lambda}$ is indeed a model based
projection. The proof is similar to that of $\lraprojf{\lambda}$.

\begin{lemma}
\label{lem:liaprojf}
$\liaprojf{\lambda}$ is a Model-Based Projection.
%Let $\mu = \liaprojf{\lambda}(M)$. Then, $\proj{M}{x} \models \mu$.
\end{lemma}

%%% Local Variables:
%%% mode: latex
%%% TeX-master: "main"
%%% End: